\documentclass[a4paper,11pt]{article}
\pdfoutput=1
\usepackage{jheppub}
\usepackage[T1]{fontenc}
\usepackage{mathtools}
\usepackage{physics}
\usepackage{amsthm}
\usepackage{amssymb}
\usepackage{braket}
\usepackage[toc,page,title]{appendix}
\usepackage{enumerate}
\usepackage[super]{nth}
\usepackage{tikz}
\usetikzlibrary{backgrounds}

\usepackage[capitalise]{cleveref}
\Crefname{lemma}{Lemma}{Lemmas}
\Crefname{definition}{Definition}{Definitions}
\Crefname{theorem}{Theorem}{Theorems}
\Crefname{corollary}{Corollary}{Corollaries}
\Crefname{section}{Section}{Sections}
\Crefname{appendix}{Appendix}{Appendices}
\Crefname{figure}{Fig.}{Figs.}
\Crefname{equation}{Eq.}{Eqs.}
\crefformat{footnote}{#2\footnotemark[#1]#3}

\DeclarePairedDelimiter\ceil{\lceil}{\rceil}
\DeclarePairedDelimiter\floor{\lfloor}{\rfloor}
\DeclareMathOperator{\poly}{poly}

\DeclareMathOperator{\BigO}{O}

\newtheorem{theorem}{Theorem}[section]
\newtheorem{lemma}[theorem]{Lemma}
\newtheorem{defn}[theorem]{Definition}

\title{Holographic duality between local Hamiltonians from random tensor networks}
\author{Harriet Apel,}
\author{Tamara Kohler}
\author{and Toby Cubitt}
\affiliation{Department of Computer Science, University College London,\\Gower St, London WC1E 6EA, UK}

\emailAdd{harriet.apel.19@ucl.ac.uk}
\emailAdd{tamara.kohler.16@ucl.ac.uk}
\emailAdd{ t.cubitt@ucl.ac.uk}

\abstract{The AdS/CFT correspondence realises the holographic principle where information in the bulk of a space is encoded at its border.
We are yet a long way from a full mathematical construction of AdS/CFT, but toy models in the form of holographic quantum error correcting codes (HQECC) have replicated some interesting features of the correspondence.
In this work we construct new HQECCs built from random stabilizer tensors that describe a duality between models encompassing local Hamiltonians whilst exactly obeying the Ryu-Takayanagi entropy formula for all boundary regions.
We also obtain complementary recovery of local bulk operators for any boundary bipartition.
Existing HQECCs have been shown to exhibit these properties individually, whereas our mathematically rigorous toy models capture these features of AdS/CFT simultaneously, advancing further towards a complete construction of holographic duality.
}


\arxivnumber{2105.12067}

\begin{document}

\maketitle
\flushbottom


\newpage
\section{Introduction}

The holographic principle states that the description of a gravitational theory in a volume can be mathematically encoded onto its lower dimensional boundary~\cite{tHooft:93,Susskind_1995}.
This principle was inspired by discussions of black hole thermodynamics, as it is thought that the information of objects inside a black hole is captured in a preserved image at the event horizon.
The most successful realisation of the holographic principle is the AdS/CFT correspondence.
AdS/CFT is a conjectured duality between quantum theories of gravity in an Anti-de-Sitter ($d+1$)-dimensional bulk and a conformal field theory (CFT) on the $d$-dimensional boundary.
It is a helpful framework to study strongly interacting quantum field theories by mapping them to semi-classical gravity in a higher dimensional space.
It is also thought that AdS/CFT will lead to insights into quantum gravity by using the duality in the opposite direction.

Quantum information is a rewarding vantage point from which to study AdS/CFT since entanglement in the correspondence has a close relationship to geometry, described by the Ryu-Takayanagi entropy formula.  Additionally, in AdS/CFT a bulk operator can be mapped to various operators on different sections of the boundary, but the operator is never cloned on two non-overlapping boundary segments. This redundancy and secrecy in the information's encoding echoes that of a quantum error correcting code, first noted in~\cite{Rindler}. Previously the tools of quantum information have been used to build toy models of the AdS/CFT correspondence~\cite{Happy, Tamara, Random, HQECC1, HQECC2, HQECC3, HQECC4, HQECC5}.
The logical degrees of freedom of a quantum code are encoded into the physical Hilbert space which can be interpreted as a bulk theory of quantum gravity being dual with a CFT on the boundary.
These mathematically rigorous toy models reproduce many of the features of holographic duality for any choice of bulk Hamiltonian, including those which do not describe any gravitational physics.
Thus, whilst they are able to reproduce aspects of the holographic dualities that arise in AdS/CFT, they do not single out gravitaional models in the bulk.
\footnote{We could choose to put a Hamiltonian which recreates aspects of gravitational physics into the bulk of these toy models. However, the emergent features of AdS/CFT captured by the models, e.g.\ complementary recovery, do not require this.}

If this type of tensor network toy model of holography is to shed light on quantum gravity, it is important to push these models further to learn where -- if anywhere -- they break down for arbitrary bulk Hamiltonians, and it becomes necessary to include gravity in the bulk in order to reproduce features of AdS/CFT.
The original tensor network models such as~\cite{Happy} reproduced the correct geometry of states and observables, but failed to map local models in the bulk to local models on the boundary. (I.e.\ Hamiltonians made up of local interactions in the bulk are mapped to Hamiltonians with no local structure on the boundary.) In light of this, one might reasonably have conjectured that to obtain a \emph{local} boundary dual model, one would need to restrict to specific bulk models. Perhaps even to models that have at least some features of bulk gravitational physics, given the non-trivial interplay between symmetries and locality involved in AdS/CFT -- and in gravitational physics more generally.~\cite{Tamara} showed this was not the case: they gave a tensor network toy model that was able to map any local bulk Hamiltonian to \emph{local} Hamiltonian on the boundary, whilst also reproducing all the same features of AdS/CFT as the original HaPPY code~\cite{Happy}.

It was known that the HaPPY code does not exactly reproduce the correct entropy scaling encapsulated in the Ryu-Takayanagi formula from AdS/CFT.~\cite{Random} improved on the original construction by showing that random tensor networks \emph{were} able to reproduce the Ryu-Takayanagi entropy scaling exactly. However, it was not known whether it was possible to construct a toy model of holographic duality that simultaneously maps between local Hamiltonians and recovers the expected Ryu-Takayangi entropy formula for general cases.

In this work, we demonstrate that the holographic toy model mapping between local Hamiltonians described by Kohler and Cubitt in~\cite{Tamara, kohler2020translationallyinvariant} can be constructed from networks of random tensors rather than tensors chosen with particular properties, thereby also reproducing the correct RyuTakayanagi entropy scaling. By showing that both these properties can be realised simultaneously, this work advances a further step along the path of mathematically rigorous constructions of holographic codes that capture more features of the AdS/CFT correspondence. Perhaps the most important consequence is to push the boundary further out between those features of holography that can already be realised without incorporating gravity into the model, and those that are inherently gravitational.

The following section of this paper gives an overview of key previous works and informally presents our main results with an overview of the proofs. \cref{Technical set-up} introduces the technical background and the notation of the relevant tensors. The full mathematical proofs of our main results are given in \cref{Results with technical details} going via results concerning the concentration of random stabilizer tensors about perfect tensors and the agreement with the Ryu-Takayanagi entropy formula, before finally presenting a description of our holographic toy model. The last section discusses the conclusions of the work and avenues for future work.

\subsection{Previous work}
\label{Previous work}

Previous work has established various holographic quantum error correcting codes (HQECC) based on tensor network structures as toy models of the AdS/CFT correspondence~\cite{Happy, Tamara, Random, HQECC1, HQECC2, HQECC3, HQECC4, HQECC5}. In a HQECC, the logical/physical code subspaces are interpreted as the bulk/boundary degrees of freedom. A complete mathematical construction of AdS/CFT is still far away, however models are increasingly capturing essential features of the duality. Among others, a successful toy model might strive to replicate holographic duality on these fronts:
\begin{enumerate}
\item \textbf{Mapping between models.} AdS/CFT is a duality between two models: the quantum theory of gravity in the bulk and a conformal field theory on the boundary.
Not only should bulk states and observables be mapped to the same on the boundary, if HQECCs are to emulate this mapping between models \textit{local} bulk Hamiltonians should map to \textit{local} boundary Hamiltonians.
\footnote{Note the bulk Hamiltonian isn't necessarily strictly local -- there may be gravitational Wilson lines which break locality in a restricted way. Our construction can also map these `quasi-local' bulk Hamiltonians to local Hamiltonians on the boundary.}
Once an encoding isometrically maps local Hamiltonians in the bulk to local Hamiltonians on the boundary, a further step would be to seek that this boundary model is Lorentz invariant and further still a quantum CFT. The reverse mapping from the boundary to the bulk is another important feature of full AdS/CFT, since it is this direction that could lead to insights into bulk quantum gravity by mapping a better understood boundary CFT.
\item \textbf{Entanglement structure.} The relationship between geometry and entropy in AdS/CFT is described by the Ryu-Takayangi formula. It states that in a holographic state the entanglement entropy of a boundary region, $A$, is proportional to the area of a corresponding minimal surface, $\gamma_A$, in the bulk geometry:
\begin{equation} \label{eqn RT}
S(A)\approx \frac{|\gamma_A|}{4G_N},
\end{equation}
here $G_N$ is Newton's constant. \cref{eqn RT} comes from classical physics in the bulk and is correct to leading order in the $G_N$ expansion. There are quantum corrections of order $G_N^0$ that come from quantum mechanical effects in the bulk, for instance the entanglement entropy between the region bounded by the minimal surface and the rest of the bulk~\cite{RT}.
\item \textbf{AdS Rindler reconstruction.} AdS/CFT has quantum error correcting features proposed in~\cite{Rindler}. The AdS-Rindler reconstruction of boundary operators from bulk operators in AdS/CFT demonstrates this property: bulk information is encoded with redundancy with complementary recovery on the boundary. On a fixed time slice, a boundary subset $A$ defines an entanglement wedge $\mathcal{E}[A]$ -- the bulk region bounded by the Ryu-Takayangi surface of $A$. The AdS-Rindler reconstruction states that for a general point in the entanglement wedge any bulk operator can be represented as a boundary operator supported on $A$~\cite{Rindler, entanglement_wedge}. Any given bulk position lies in many entanglement wedges with distinct associated boundary subregions, hence there are multiple representations for a single bulk operator with different spatial support on the boundary. Since the bulk operator can be reconstructed on $A$ is it protected against an error where the complementary subregion $\bar{A}$ is erased.
Given any partition of the boundary into non-overlapping regions $A$ and $\bar{A}$, where the union of $\mathcal{E}[A]$ and $\mathcal{E}[\bar{A}]$ cover the entire bulk spatial slice, a given bulk operator should always be recoverable on exactly one of the regions (a property known as \emph{complementary recovery}).
\footnote{In AdS/CFT it is not guaranteed that the union of $\mathcal{E}[A]$ and $\mathcal{E}[\bar{A}]$ will cover the entire bulk spatial slice, for example if there are horizons in the bulk. However if there is no bulk entanglement -- a regime we will restrict to for quantitative study of our toy models -- this condition is satisfied and we would expect complementary recovery.}\end{enumerate}
We will briefly describe three notable HQECCs based on tensor network constructions that this work will draw heavily upon, outlining their successes and limitations with respect to the above three points.

The HaPPY code~\cite{Happy} was the first explicit construction of a HQECC. The encoding has a tensor network structure where tensors are arranged in a pentagonal tiling of hyperbolic 2-space, depicted in \cref{fg holographic code}. A particular choice of isometric tensor was selected in order for the total encoding to be a bulk-to-boundary isometry, able to map states and observables with no loss of information.
However, a local Hamiltonian in the bulk, given by $H = \sum_z h_z \otimes \mathbb{I}$ where $h_z$ are local Hermitian observables, is generally mapped to a non-local Hamiltonian on the boundary with global interactions.
By including global interactions the Hamiltonian has lost its connection to the boundary geometry and it is not meaningful to describe it as having the dimensionality of the boundary. The HaPPY code is a good model of the AdS-Rindler reconstruction since a given bulk operator can be mapped to non-unique boundary operations, all with the same action on the code subspace. Broken symmetry of the boundary through discretisation does however introduce pathological boundary regions, whereby complementary recovery is violated and the bulk operator is not recoverable on either $A$ or $\bar{A}$. The entanglement structure of the model echoes that of holographic duality with the Ryu-Takayanagi formula obeyed exactly in special cases where $A$ is a connected boundary region and there are no bulk degrees of freedom. Yet for uncontracted bulk indices Ryu-Takayanagi is not obeyed and only entropy bounds are manifest. Furthermore, if the bulk input state is entangled even these bounds do not hold. While imperfect, these codes do capture key holographic properties making them an important footing for further work.

\begin{figure}[tbp]
\centering
\includegraphics[trim={0cm 0cm 0cm 0cm},clip,scale=0.34]{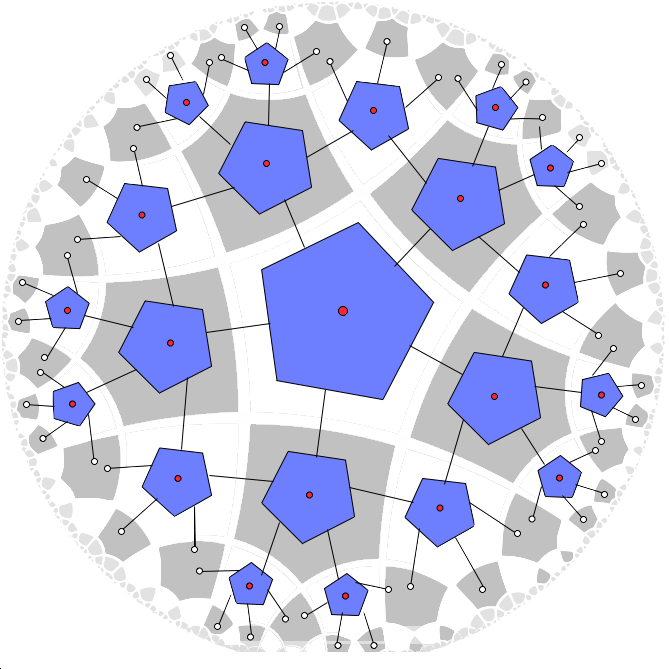}
\caption{The holographic pentagon code from~\cite{Happy}. Each blue pentagon represents a 6 leg perfect tensor (see \cref{defn Perfect tensors}) with legs partitioned to represent 5 physical output legs and 1 logical input leg (shown by the red dot).}
 \label{fg holographic code}
\end{figure}

An extension of the HaPPY code was proposed in~\cite{Tamara} that applied the recent theory of Hamiltonian simulation from~\cite{simulation}. The same class of isometric tensors as the HaPPY code is used, this time choosing tensors describing qudit rather than qubit systems. The tensor network geometry is also revised from~\cite{Happy} with a 3-dimensional bulk tessellated by Coxeter polytopes. This generalisation to higher dimensions enabled techniques from Hamiltonian complexity theory to be employed to break down global interactions, so that local Hamiltonians in the bulk now map to local Hamiltonians on the boundary.
The mapping described is therefore between models, however it should be noted that the boundary model doesn't exhibit any discrete version of Lorentz invariance or conformal symmetry.\footnote{In~\cite{HQECC2} a holographic construction where the boundary theory is invariant under Thompson's group, a discrete analogue of the conformal group, has been constructed.}
\cite{Tamara}'s construction was able to surpass a static bulk and boundary to explore how the model's dynamics qualitatively mirrors features of AdS/CFT. This work also completed the dictionary with the reverse boundary-to-bulk map which facilitates an insight into the connection between bulk and boundary energy scales. The error correction properties and entanglement structure are inherited unchanged from the underlying HaPPY-like code.
In~\cite{kohler2020translationallyinvariant} this was extended to a 2D-1D holographic mapping which mapped (quasi)-local bulk Hamiltonians to local boundary Hamiltonians.

\cite{Random} studied toy models of AdS/CFT based on random tensor networks with unconstrained graph geometries in $\geq 2$ dimensions. In the large bond dimension limit, random tensor networks are approximate isometries from the boundary to the bulk and in the reverse direction, defining bidirectional holographic codes. As with the HaPPY code, the encoding does not map between local Hamiltonians, instead producing global interactions. The bulk-to-boundary mapping does however satisfy the error correction properties of the AdS-Rindler reconstruction. The main triumph of this model is that the entanglement entropy of all boundary regions obey the Ryu-Takayanagi formula with the expected corrections when there is non-trivial quantum entanglement in the bulk. This natural likeness between the entanglement structure of high-dimensional random tensor networks and holography might suggest that there is a deeper link between semi-classical gravity and scrambling/chaos.

Guided by these results we will work towards constructing a HQECC that simultaneously exhibits both the local Hamiltonian bulk-boundary correspondence of~\cite{Tamara} and the Ryu-Takayanagi agreement of~\cite{Random}. More details on the relevant proof techniques used in these works are discussed in \cref{Results with technical details}.

\subsection{Our results}
\label{Our results}

We set up a duality between states, observables and local Hamiltonians in $d$-dimensional hyperbolic space and its $(d-1)$-dimensional boundary, for $d=2,3$. In the model a general local Hamiltonian acting in the bulk has an approximate dual 2-local nearest-neighbour Hamiltonian on the boundary. The mapping has redundant encoding leading to error correcting properties where the reconstruction of any bulk operator acting in the entanglement wedge of a boundary region is protected against erasure of the rest of the boundary Hilbert space. This implies complementary recovery for all partitions of the boundary and all local bulk operators. The entanglement entropy of general boundary regions obeys the Ryu-Takayanagi formula exactly where there is no bulk entanglement. Furthermore the effect of introducing bulk entanglement qualitatively agrees with the entropic corrections expected in real AdS/CFT.

The explicit encoding is a chain of simulations, the first of which is described by a tensor network HQECC. The geometry of our network is inherited from~\cite{Tamara} where hyperbolic bulk space is tessellated by space-filling Coxeter polytopes. In our set-up a random stabilizer tensor is placed in each polytope with one tensor index identified as the bulk index and the rest contracted through the faces of the polytopes. With a suitably high tensor bond dimension we are able to achieve several features of AdS/CFT simultaneously with high probability. Therefore, by choosing a model of semi-classical gravity in the bulk this construction is an explicit toy model of holography, providing a mathematically rigorous tool for exploring the physics of this setting. The notable features of our construction are summarised in the following statement of our main result, these are then made precise in \cref{holography} for the 2D to 1D duality and in \cref{thm Main result} for 3D to 2D.

\begin{theorem} [Informal statement of holographic constructions, \cref{holography} and \cref{thm Main result}] \label{informal theorem}
Given any (quasi) local bulk Hamiltonian acting on qudits in ($d=2,3$)-dimensional AdS space, we can construct a dual Hamiltonian acting on the $(d-1)$-dimensional boundary surface such that:
  \begin{enumerate}
  \item All relevant physics of the bulk Hamiltonian are arbitrarily well approximated by the boundary Hamiltonian, including its eigenvalue spectrum, partition function and time dynamics.
  \item The boundary Hamiltonian is 2-local acting on nearest neighbour boundary qudits, realising a mapping between models.
  \item The entanglement structure mirrors AdS/CFT as the Ryu-Takayanagi formula is obeyed exactly for any boundary subregion in the absence of bulk entanglement.
  \item Any local observable/measurement in the bulk has a set of corresponding observables/measurements acting on portions of boundary as described by the AdS Rindler reconstruction giving complementary recovery.
  \end{enumerate}
\end{theorem}

There are two main steps to proving this result.
The first is using measure concentration techniques to demonstrate that, with high probability, a HQECC constructed out of high dimensional random stabilizer tensors will be a an isometry from bulk to boundary, and will exhibit complementary recovery and obey the Ryu-Takayanagi formula.
The next step is to use Hamiltonian simulation techniques to demonstrate that the boundary Hamiltonian which results from pushing a local bulk Hamiltonian through the random tensor network can be simulated arbitrarily well by a 2-local boundary Hamiltonian.
The level of approximation depends on parameters in the local boundary Hamiltonian -- with the most significant parameter being maximum weight of interactions in the Hamiltonian.

This result could be extended to higher dimensions.
In particular, all the techniques in \cref{thm Main result} generalise to higher dimensions, all that is needed to generalise the result is to find examples of tessellations meeting the requirements of the construction.
While the formal theorems set out the details of the boundary, intuitively the geometry is what we expect from the concept of a boundary of a finite bulk space and the number of boundary qudits scales almost linearly with the number of bulk qudits.

\subsection{Overview of methodology}

As outlined above, we first build a tensor network out of random tensors and demonstrate that it is a mapping from bulk to boundary spins which captures key features of AdS/CFT.
Then we apply Hamiltonian simulation techniques to the boundary to construct a local boundary  model.

The Hamiltonian simulation techniques we use require additional boundary qudits in order to achieve a local boundary Hamiltonian.
Non-increasing Pauli rank of the operator as it is pushed through the tensor network is essential for this step of the simulation in order to achieve a reasonable scaling of the final set of boundary qudits with the bulk qudits.
Therefore, we cannot apply Hamiltonian simulation techniques to a HQECC constructed from Haar random tensors.
We instead use random stabilizer tensors which are generated by the Clifford group, forming a 2-design on prime dimensional qudits.
Since our construction involves mapping the bulk Hamiltonian through the tensor network to the boundary, we require the network formed from these random stabilizer tensors to be an isometry.

It follows from measure concentration that random tensors with large bond dimension describe approximate isometries with high probability.
In the case of random stabilizer tensors, this can be strengthened to saying that they are \emph{exactly} perfect with high probability, using quantisation of entropy results.
By increasing the bond dimension this probability can be pushed arbitrarily close to~1, which is crucial because it means that if we construct a large tensor network, and pick each tensor uniformly at random, we can efficiently construct a network where each random stabilizer tensor is perfect.
The probability of the network being an isometry is then given by the union bound of all $n$ tensors in the network being perfect, \cref{thm Random stabilizer tensors are perfect}.
Hence we show that for any sized network, one can choose a sufficiently large prime $p=O(n^q)$ bond dimension (where $q>1/b$ and $0<b<1$), such that the probability of obtaining an isometry is close to 1.

The next step in our construction involves approximate simulations.
It is the Hamiltonian parameters in this step that determine how accurately the resulting boundary Hamiltonian reproduces the relevant physics of the bulk Hamiltonian stated in point~1 of \cref{informal theorem}.
By tuning these Hamiltonian parameters we can make this step arbitrarily accurate.
The technique for carrying out the approximate simulation depends on the spatial dimension of the HQECC.
In the 2D-1D case, we use the `history state simulation method' from~\cite{kohler2020translationallyinvariant}.
For this simulation technique, increasing the accuracy of the simulation requires increasing the number of ancilla qudits and increasing the weight of terms in the Hamiltonian.
For the 3D-2D case (or higher), we can use perturbative simulations built out of perturbation gadgets~\cite{gadgets,simulation}.
In this case it is just the weight of terms in the Hamiltonian which determines the accuracy of the approximation.

The entanglement structure and complementary recovery is proven via a technique introduced in~\cite{Random}, relating functions of the random tensor network to partition functions of the classical Ising model.
The error in approximating the partition function with the well-known Ising model ground state is suppressed polynomially with the bond dimension of the tensor, which we have already chosen to be large in order to achieve an isometric encoding.
Using quantisation of entropy for stabilizer states, we can show that for large enough bond dimension a random stabilizer tensor network with a Coxeter polytope structure obeys the Ryu-Takayanagi formula exactly for general boundary regions when contracted with a bulk stabilizer state.
To generalise this to bulk product states, we first prove complementary recovery when the Coxeter polytopes generating the honeycombing of the bulk are odd-faced.
This allows us to convert bulk stablilizer states into general product states while retaining the Ryu-Takayanagi entropy statement, \cref{lm Exact Ryu-Takayanagi}.
In the 2D/1D construction, the final simulation step  does introduce small errors in the Ryu-Takayanagi entropy formula.
These are polynomially suppressed by the high-weight interaction terms in the boundary Hamiltonian, and can be made arbitrarily close to zero.
In the 3D/2D construction, the ancilla qudits do not change the entanglement structure, so do not introduce errors in the Ryu-Takayanagi formula.

The proof combines many of the techniques from previous works: perturbation gadgets and the structure of the Coxeter group~\cite{Tamara}; the theory of Hamiltonian simulation~\cite{simulation}, history state simulation techniques~\cite{kohler2020translationallyinvariant} and drawing comparison to a classical Ising model~\cite{Random}.
The main function of our result is to prove that these techniques can be successfully employed together with no arising contradictions, generating a more complete toy model of holography. Full proof of the result and the relation between the bond dimension and the probability of obtaining the encoding described are given in \cref{Results with technical details}.

\section{Technical set-up}
\label{Technical set-up}

A pure quantum state on $t$ qudits of prime dimension $p$ is described by a rank $n$-tensor with components $T_{i_1i_2...i_t}$, where each index runs over $p$ values:
\begin{equation} \label{eqn tensor state}
\ket{\psi} = \sum_{i_1i_2...i_t \in \mathbb{Z}_p^t} T_{i_1i_2...i_t} \ket{i_1}\otimes \ket{i_2} \otimes ... \otimes \ket{i_t}.
\end{equation}
Consider bipartitioning this Hilbert space so that $\mathcal{H}_T = \mathcal{H}_A \otimes \mathcal{H}_B$ where $d_A \leq d_B$. Grouping the tensor's indices together into these two subsystems a state on the total Hilbert space is given by
\begin{equation}
\ket{\psi} = \sum_{a,b} T_{a,b} \ket{a} \otimes \ket{b}.
\end{equation}
By reshaping the tensor to consider $t_A$ `input legs' and $t_B$ `output legs' it represents a linear map, $V$, between two Hilbert spaces $\mathcal{H}_A \mapsto \mathcal{H}_B$ with dimensions $d_A=p^{t_A}$ and $d_B=p^{t_B}$ respectively,
\begin{equation} \label{eqn tensor mapping}
V\ket{a} = \sum_b T_{a,b} \ket{b}.
\end{equation}

Larger mappings can be generated by contracting multiple tensors together to form a tensor network; in holography this mapping is from states and observables acting on bulk indices to those acting on the boundary and vice versa. In this way tensor networks are a useful tool to represent HQECCs, since through the graphical notation the structure and properties of the encoding are manifest. Visual depictions of the mathematical objects described in the above equations are shown in \cref{fg tensor diagrams}. The choice of tensor and network geometry in quantum-code-based toy models of holography ultimately determines the model's properties and in what ways it successfully mimics AdS/CFT.

\begin{figure}[tbp]
\centering
\includegraphics[trim={13cm 12cm 2cm 4cm},clip,scale=0.28]{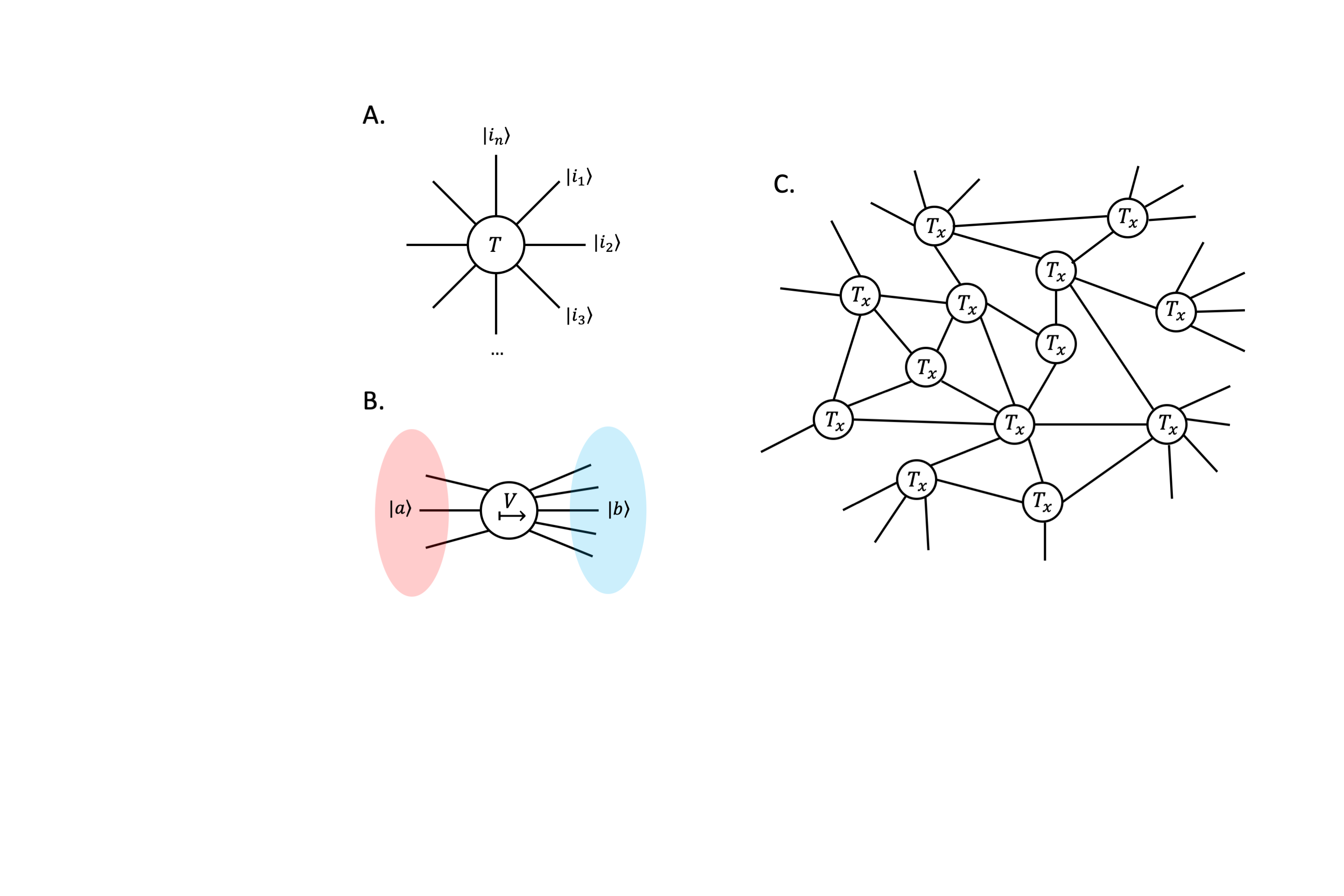}
\caption{A) Graphical representation of tensor from \cref{eqn tensor state}. B) Graphical representation of reshaping a tensor into a mapping from \cref{eqn tensor mapping}. C) Example of contracting tensors together to form a tensor network. }
 \label{fg tensor diagrams}
\end{figure}

\subsection{Perfect tensors}
\label{Perfect tensors}

The properties of the state describing the tensor, in the sense of \cref{eqn tensor state}, will determine the tensor's properties.
The set of maximally entangled states (AME states) inspired the classification of a set of isometric tensors, called `perfect tensors' in the context of the HaPPY code construction~\cite{Happy}\footnote{This idea has been extended to tensors with $2m+1$ indices via pseudo-perfect tensors, Definition 5.0.2 in~\cite{Tamara}.}.
AME sates are maximally entangled for all possible bipartitions, however there are other equivalent defining properties:
\begin{defn} [Absolutely maximally entangled states, Definition 1 from~\cite{AME}]
\label{defn AME}
\textit{An absolutely maximally entangled state is a pure state, shared among $t$ parties $P = \{1, . . . , t\}$, each having a system of dimension $p$. Hence $\ket{\Phi}\in \mathcal{H}_1 \otimes ... \otimes \mathcal{H}_t$, where $\mathcal{H}_i \cong \mathbb{C}^p$, with the following equivalent properties:
\begin{enumerate}[(i)]
\item  $\ket{\Phi}$ is maximally entangled for any possible bipartition. This means that for any bipartition of $P$ into disjoint sets $A$ and $B$ with $A \cup B = P$ and, without loss of generality,  $m = |B| \leq |A| = t-m$, the state $\ket{\Phi}$ can be written in the form
\begin{equation}
\ket{\Phi} = \frac{1}{\sqrt{p^m}} \sum_{k\in \mathbb{Z}_p^m}\ket{k_1}_{B_a} ... \ket{k_m}_{B_m}\ket{\phi(k)}_A,
\end{equation}
with $\braket{\phi(k)|\phi(k')}=\delta_{kk'}$.
\item The reduced density matrix of every subset of parties $A \subset P$ with $|A| = \floor*{\frac{t}{2}} $ \footnote{The notation$\floor*{x}$ corresponds to rounding down to the nearest integer.} is totally mixed, $\rho_A = p^{-\floor*{\frac{t}{2}}}\mathbb{I}_{p^{-\floor*{\frac{t}{2}}}}$.
\item The reduced density matrix of every subset of parties $A \subset P$ with $|A| \leq \frac{t}{2}$ is totally mixed.
\item The von Neumann entropy of every subset of parties $A \subset P$ with $|A|  =  \floor*{\frac{t}{2}} $ is maximal, $S(A) = \floor*{\frac{t}{2}}  \log p$.
\item The von Neumann entropy of every subset of parties $A \subset P$ with $|A| \leq \frac{t}{2}$ is maximal, $S(A) = |A| \log p$.
\end{enumerate}
These are all necessary and sufficient conditions for a state to be absolutely maximally entangled.
We denote such a state as an $AME(t, p)$ state.}
\end{defn}

A perfect tensor acts as an isometric mapping from any operator acting on less than half of its legs to a new operator acting on the remaining legs.
Formally,
\begin{defn}[Perfect tensors, Definition 2 from~\cite{Happy}]
\label{defn Perfect tensors}
\textit{A 2$m$-index tensor $T_{a_1a_2...a_{2m}}$ is a perfect tensor if, for any bipartition of its indices into a set $A$ and complementary set $\bar{A}$ with $|A| \leq |\bar{A}|$, $T$ is proportional to an isometric tensor from $A$ to $\bar{A}$.}
\end{defn}
Examining the necessary and sufficient condition $(iii)$ from \cref{defn AME} reveals the connection between AME states and perfect tensors.
$\rho_A$ can be calculated from the reshaped tensor:
\begin{subequations}
\begin{align}
\rho_A &= \trace_B\left[ \ket{\psi}\bra{\psi} \right]\\
& = \trace_B \left[\sum_{a,b,a',b'} T_{a,b}T_{a',b'} (\ket{a}\otimes \ket{b}) (\bra{a'}\otimes \bra{b'}) \right]\\
&= \sum_{a,a',b} T_{a,b}T_{a',b}^\dagger \ket{a}\bra{a'}= VV^\dagger.
\end{align}
\end{subequations}
Therefore $\ket{\psi}$ is an AME state if and only if the tensor $T$ is perfect in the sense of \cref{defn Perfect tensors}.

Perfect tensors are a convenient choice in model holographic constructions since they have many useful properties.
For the physics on the boundary to be dual to that in the bulk, no information can be lost and the encoding map between the two spaces must be an isometry.
Since the product of isometries is another isometry, a tensor network comprising of perfect tensors is an isometric mapping -- providing the network geometry is such that the input of any tensor acts on the minor partition.
This is the case in the HaPPY code and in the generalisation proposed in~\cite{Tamara}.
Perfect tensors also produce desirable error correcting properties since every AME state is the purification of a quantum maximum distance separable (MDS) code (Theorem B.1 in~\cite{Tamara}).
Such a code saturates the quantum Singleton bound, a constraint on the distance of a $[[n,k,d]]$ quantum code:
\begin{equation}
n-k \geq 2(d-1).
\end{equation}
In a MDS code the entire code space is required to accommodate all the correctable errors and there is no wasted space.

\subsubsection{Perfect stabilizer tensors}
\label{Perfect stabilizer tensors}

Perfect \textit{stabilizer} tensors that describe \textit{stabilizer} AME states are particularly suited to the construction of quantum error correction codes. Given a factorisation of a Hilbert space into $t$ $p$-dimensional qudits. A stabilizer state on this space is the unique simultaneous $+1$ eigenvector of its stabilizer $S$,
\begin{equation}
\{\, \ket{\psi}\, | \,P\ket{\psi} = \ket{\psi} \, ,\forall \,P \in S \, \}.
\end{equation}
$S$ is an Abelian subgroup of the Pauli group excluding the element $\omega I$. In order to specify a unique state the stabilizer's minimal generating set must contain $t$ elements so that the total size of the subgroup is $|S|=p^t$.

The constructions of both~\cite{Happy} and~\cite{Tamara} are based on perfect stabilizer tensors. These tensors inherit the desirable characteristics of perfect tensors with some additional useful properties. The family of maximum distance QECCs described by different partitions of these tensors are \textit{stabilizer} codes where the stabilizer is derived from the stabilizer of the state (Theorem D.1 in~\cite{Tamara}). The group theoretic structure of stabilizer codes fosters other useful properties. A notable example is that for a perfect stabilizer tensor there exists a consistent basis for the family of QECCs that maps logical Pauli operators to physical Pauli operators hence preserving Pauli rank (Theorem D.4 in~\cite{Tamara}). The importance of this property in assembling an increasingly representative toy model is discussed in \cref{HQECC between local Hamiltonians with random stabilizer tensors}.

\subsection{Random tensors}
\label{Random tensors}

Random tensors can be generated via random states on the respective Hilbert space. To obtain the random state $\ket{\phi} = U\ket{0}$, start from an arbitrary reference state, $\ket{0}$, and apply a random unitary operation, $U$. The average over a function of the random state, $f(\ket{\phi})$, is given by integration over the unitary group, $U$, with respect to the Haar measure
\begin{equation}
\langle f(\ket{\phi}) \rangle = \int_{\mathcal{U}(d)} f(\ket{\phi}) dU.
\end{equation}
The Haar probability measure is a non-zero measure $\mu$ such that if $h$ is a probability density function on the group $G$, for all $S\subseteq G$ and $g\in G$:
\begin{equation}
\mu(gS) = \mu(Sg) = \mu(S),
\end{equation}
where
\begin{equation}
\mu(S)  := \int_{g\in S} d\mu (g) = \int_{g\in S} h(g) dg, \qquad \mu(G)  := 1.
\end{equation}
A unique Haar measure exists on every compact topological group, in particular the unitary group~\cite{Haar}.

A random tensor does not generally have any particular properties. It is only in the limit of large bond dimension of the tensor, i.e. high dimensional qudits, where networks of random tensors have been shown to exhibit holographic properties, particularly entanglement structure~\cite{Random}.

\subsubsection{Random stabilizer tensors}
\label{Random stabilizer tensors}

Random \textit{stabilizer} tensors are analogously generated by uniformly choosing \textit{stabilizer} states at random. In this case the reference state is chosen as a stabilizer state $\ket{\tilde{\psi}}$, stabilized by $S$, and instead of a random unitary, a random Clifford unitary, $C$, is applied to generate the random stabilizer state $\ket{\psi} =C \ket{\tilde{\psi}} $. Since elements of the Clifford group map the Pauli group to itself under conjugation, the resulting state is stabilized by $S' = C S C^\dagger$:
\begin{subequations}
\begin{align}
C P C^\dagger \ket{\psi} & = C P C^\dagger C \ket{\tilde{\psi}}\\
& =  C P \ket{\tilde{\psi}}\\
& =  C \ket{\tilde{\psi}}\\
& = \ket{\psi}
\end{align}
\end{subequations}
In the case of qudits of prime dimension the same procedure is followed for generating random stabilizer tensors, substituting for the generalised Pauli and Clifford operators.

The unitary Haar measure described in the above section is invariant under arbitrary unitary transformations. While uniformly sampling over the Clifford group is not equal to the Haar measure on unitary groups, we can still exploit the invariance of the Haar measure through designs:
\begin{defn}[Unitary 2-design,~\cite{unitary_design}]
\label{defn Unitary 2-design}
A set $\mathcal{D} = \{U_k\}_{k=1}^K$
of unitary matrices on $\mathcal{H}=\mathbb{C}^d$  is a unitary 2-design if it fulfils the condition:
\begin{enumerate}
\item (Twirling of states) For all $\rho\in \mathcal{B}(\mathcal{H}\otimes \mathcal{H})$
\begin{equation} \label{eqn Twirling of states}
\frac{1}{K}\sum_{U_k \in \mathcal{D}} \left(U_k \otimes U_k \right) \rho \left(U_k \otimes U_k \right)^\dagger = \int_{\mathcal{U}(d)} (U\otimes U)\rho (U\otimes U)^\dagger dU.
\end{equation}
\end{enumerate}
\end{defn}
Uniform distribution over the Clifford group on qubits is known to be a unitary 2-design, hence the \nth{2} moment of the ensemble is equal to the \nth{2} moment of the invariant Haar-random unitary~\cite{unitary_design, Clifford_2-design}. This result was generalised to Cliffords acting on qudit systems of prime power dimensions in~\cite{Clifford_2-design_qudit}. It is a network of these random stabilizer tensors that we will base our HQECC on and the fact that the Clifford group generates stabilizer states and is simultaneously a 2-design will be key in the following results.

\subsection{Hyperbolic Coxeter groups}\label{sect Hyperbolic Coxeter groups}

Our tensor network comprises of a tensor network living in hyperbolic space.
For $d>2$ analysing properties of hyperbolic tessellations by visualisation becomes cumbersome so a systematic approach is required.
\cite{Tamara} has already done this heavy lifting, generating tessellations of hyperbolic space using Coxeter polytopes which can be analysed using their associated Coxeter system.
We will recycle unchanged this procedure and the scalings that arise from such tessellations, so we refer to~\cite{Tamara} for details but include a short summary here of the main definitions.

\begin{defn}[Coxeter system,~\cite{coxeter1}]
\label{defn Coxeter system}
Let $S=\{s_i\}_{i\in I}$ for $I \subset \mathbb{Z}$ be a finite set.
Let $M=(m_{i,j})_{i,j\in I}$ be a matrix such that:
\begin{itemize}
\item $m_{ii} = 1$, $\forall i \in I$
\item $m_{ij}=m_{ji}$, $\forall i,j \in I, i\neq j$
\item $m_{ij}\in ( \mathbb{N} \setminus \{ 1\})\cup \{ \infty\}$,  $\forall i,j, \in I, i\neq j$
\end{itemize}
$M$ is called the Coxeter matrix.
The associated Coxeter group, $W$, is defined by the presentation:
\begin{equation}
W = \langle S | (s_is_j)^{m_{i,j}}= 1 \quad \forall i,j \in I \rangle
\end{equation}
The pair $(W,S)$ is called a Coxeter system
\end{defn}

Each Coxeter system has an associated polytope, $P$.
The $P$ tiles in $\mathbb{X}^d$ and the reflections in the facets of $P$ generate a Coxeter matrix $(m_{ij})_{i,j \in I}$ which is a discrete subgroup of $Isom(\mathbb{X}^d)$~\cite{coxeter3}.

\begin{defn}[Coxeter polytope,~\cite{coxeter2}\cite{Tamara}]
\label{defn Coxeter polytope}
A convex polytope in $\mathbb{X}^d = \mathbb{S}^d, \mathbb{E}^d$ or $\mathbb{H}^d$ is a convex intersection of a finite number of half spaces.
A Coxeter polytope $P\subseteq \mathbb{X}^d$ is a polytope with all dihedral angles integer submultiples of $\pi$.
\end{defn}

Properties of the Coxeter system will determine properties of the resulting tessellation and therefore the HQECC embedded in it.
Hyperbolic Coxeter groups have a growth rate $\tau >1$, which determines the number of the number of polyhedral cells in a ball of radius $r$ scaling as $O(\tau^r)$.
This is used to bound the final number of boundary qudits required.
For the resulting network to describe a bulk to boundary isometry when we place a tensor in each polyhedral cell, it is essential that the number of output indices is at least as large as the number of output indices (including the bulk index).
Since we are working in negative curvature space, most polytopes in the tessellation will have more facets shared with polytopes in the next layer (at a larger radius) than with the previous layer (at a smaller radius).
While this property is not generally true of all tessellations generated by Coxeter polytopes, Theorem 6.1 of~\cite{Tamara} gives a condition on the Coxeter system that is sufficient to ensure this is always the case.
These are just two examples of where analysing the Coxeter groups gives a systematic method of choosing an appropriate tessellation for HQECCs

\section{Results with technical details}
\label{Results with technical details}

We construct tensor network HQECCs by embedding random stabilizer tensors with $t$ legs into each cell of space-filling tessellations of hyperbolic 2-space and hyperbolic 3-space.
The tessellations are defined by Coxeter systems chosen such that the Coxeter polytope has an odd number of faces\footnote{An odd number of faces is essential for complementary recovery see \cref{lm Complementary recovery}}.
Each tensor has one uncontracted leg associated as the logical bulk index, the remaining legs are contracted with the tensors that occupy the neighbouring polyhedral cells.
The tessellation is finite such that at the cut-off the uncontracted tensor legs become the physical boundary indices.

In this section we demonstrate particular properties of this construction.
First we present a mathematically rigorous characterisation of the concentration of random stabilizer tensors about perfect tensors with increasing bond dimension, using the algebraic structure of the stabilizer group to arrive at a probability bound on having an exact perfect tensor.
Then we demonstrate via an Ising model mapping that the entanglement structure of general boundary subregions obeys the Ryu-Takayanagi formula exactly when there is no entanglement in the bulk input state.
Here we also show that the construction exhibits complementary recovery for all choices of boundary bipartition and all local bulk operators, an advance on previous models where exceptions could be found.
Finally we use simulation techniques to break down global interactions at the boundary to demonstrate that at the same time as achieving exact Ryu-Takayanagi we can describe a duality between models that encompasses local Hamiltonians.

\subsection{Random stabilizer tensor networks describe an isometry}
\label{Random stabilizer tensor networks describe an isometry}

In contrast to the perfect tensor case, it is not immediately clear that random stabilizer tensor networks with large dimension correspond to an \textit{exact} isometric mapping or maximum distance stabilizer codes.
Since the product of isometries is still an isometry, we can focus on proving a result for an individual random stabilizer tensor.
Our first result is that random stabilizer tensors with large bond dimension are perfect stabilizer tensors with high probability.
This implies that simultaneously exhibiting the advantageous properties of perfect stabilizer and random tensors is realisable.
This was previously shown to be true for random tensors in HQECCs where the dimensions of the bulk indices and boundary indices were chosen appropriately~\cite{Random}.
In the following we explicitly work out the relation between the probability and bond dimension, and our result is applicable to tensors with uniform bond dimension.

The key ingredient in this proof is that random states in high dimensional bipartite systems are subject to the `concentration of measure' phenomenon.
This means that with high probability a function of the random state will concentrate about its expectation value~\cite{Aspects_of_generic_entanglement}.
We will show that for a random stabilizer state on $n$ qudits, the von Neumann entropy of every subset of the tensor's indices with $|A|= \floor*{t/2}$ is maximal with high probability.
This is a necessary and sufficient condition for the state to be AME and therefore for the tensor to be perfect. To ensure the tensor is exactly perfect opposed to close to perfect -- which would not guarantee that the tensor described MDS stabilizer codes -- we will use the particular algebraic structure of stabilizer states to constrain the values that the entropy can take.

\subsubsection{Concentration bounds}
\label{Concentration bounds}

Measure concentration is the surprising observation that a uniform measure on a hypersphere will strongly concentrate about the equator as the dimension of the hypersphere grows. This implies that a smoothly varying function of the hypersphere will also concentrate about its expectation. Levy's lemma formalises the concentration of measure in the rigorous sense of an exponential  probability bound on a finite deviation from the expectation value:
\begin{lemma}[Levy's lemma; see~\cite{Levys_lemma}]
\label{lm Levy's lemma}
Given a function $f:\mathbb{S}^d \mapsto R$ defined on the $d$-dimensional hypersphere $\mathbb{S}^d$, and a point $\phi \in \mathbb{S}^d$ chosen uniformly at random,
\begin{equation}
\textup{Prob}_H \left[ |f(\phi)-\langle f \rangle | \geq \epsilon\right] \leq 2 \exp \left( \frac{-2C_1(d+1)\epsilon^2}{\eta^2} \right),
\end{equation}
where $\eta$ is the Lipschitz constant of $f$, given by $\eta = \sup |\nabla f |$, and $C_1$ is a positive constant (which can be taken to be $C_1=(18\pi^3)^{-1}$).
\end{lemma}

Pure $d_{AB}$-dimensional quantum states can be represented by points on the surface of a hypersphere in  $(2d_{AB}-1)$ dimensions due to normalisation.
Therefore by setting $d=2d_{AB}-1$, the above can be applied to functions of states $\ket{\phi}_{AB}$ chosen randomly with respect to the Haar measure.
However, for this construction it is important that we use random stabiliser states, $\ket{\psi}_{AB}$, so we must take an exact 2-design. Low showed in~\cite{Low_2009} that in general t-designs give large deviations, particularly for low t.
By leveraging intermediate results of~\cite{Low_2009} alongside a quantisation condition on entropy for stabiliser states we will show that $S(\trace_B(\ket{\psi}\bra{\psi}))=\log d_A$ with high probability.  In order to arrive at this more general concentration result for entropy we will need the following bound on moments.

\begin{lemma}[Bound on moments; see~\cite{Low_2009} Lemma 3.3]
\label{lm moment bound}
Let $X$ be any random variable with probability concentration
\begin{equation}
\textup{Prob} \left(|X-\mu|\geq \epsilon \right) \leq C_2e^{-a\epsilon^2}.
\end{equation}
(Normally $\mu$ will be the expectation of $X$, although the bound does not assume this.)
Then,
\begin{equation}
\langle |X-\mu|^m \rangle  \leq C_2 \Gamma(m/2 +1)a^{-m/2} \leq C_2 \left( \frac{m}{2a}\right)^{m/2}
\end{equation}
for any $m>0$.
\end{lemma}

In particular, if the function $f$ is a polynomial of degree 2 then the expectation value with respect to the Haar measure or an exact 2-design are equal. While the von Neumann entropy is not such a function, the flat eigenspectra of stabiliser states demonstrated in \cref{Quantised entropy of stabilizer states}  allows us to equivalently consider the R\'enyi-2 entropy which is the logarithm of a degree 2 polynomial, $S_2(\rho_A)= -\log(\trace(\rho_A^2))$. We first use Levy's lemma to bound the probability concentration of the purity of Haar random states, then using \cref{lm moment bound} link this concentration to entropy of pseudo random states.

\subsubsection{Expectation of $\trace(\rho_A^2) $}
\label{Expectation of purity}

The average reduced density matrix of any bipartite Haar random state is the maximally mixed state, $\mathbb{I}/d_A$, following from the invariance of the Haar measure. This is the premise for expecting that the average entropy of high dimensional random stabilizer states is close to maximal with some small fluctuation.
However, we will need the expectation of the purity of random stabilizer states, $\trace(\rho_A^2)$, in order to translate Levy's lemma applied to a Haar random degree-2 polynomial into a concentration statement concerning the entropy of stabilizer states.

We start by proving that the expectation of a tensor product of two copies of the random stabiliser state density matrix is proportional to the projector onto the symmetric subspace.
This result is a specific case of Proposition 6 from~\cite{sym_subspaces} suitable when $n=2$ and $\rho$ is a 2-design.
We will use a known theorem that relates the symmetric subspace to representation theory:

\begin{theorem}[Symmetric subspaces; see~\cite{sym_subspaces} Theorem 5]
\label{thm sym subspaces}
For $U$ the $d$\nobreakdash-dimensional unitary group, $U^{\otimes n}$ acts as an irreducible representation on the symmetric subspace of $(\mathbb{C}^d)^{\otimes n}$.
\end{theorem}

\begin{lemma}[Average of $\rho^{\otimes 2}$]
\label{lm Average of rho}
Let $\ket{\psi}\in \mathcal{H}_{AB}$ with dimension $d_{AB}$ be a random stabilizer state obtained by applying a random element of the Clifford group $\{U_k \}_{k=1}^K$ to a reference stabilizer state $\ket{\tilde{\psi}}_{AB}$. Given the density matrix $\rho = \ket{\psi} \bra{\psi}$, the average over two copies of this density matrix is given by:
\begin{equation}
\bigg\langle \ket{\psi}\bra{\psi}_{AB}\otimes \ket{\psi}\bra{\psi}_{A'B'}\bigg\rangle = \frac{\mathbb{I}_{ABA'B'}+\mathcal{F}_{ABA'B'}}{d_{AB}(d_{AB} +1)},
\end{equation}
where $\mathbb{I}_{ABA'B'}$ is the identity of the Hilbert space $\mathcal{H}_{AB} \otimes \mathcal{H}_{A'B'}$ and $\mathcal{F}_{ABA'B'}$ is the swap operator \footnote{$\mathcal{F}_{KK'}(\ket{\phi}_K \otimes \ket{\psi}_{K'}) = \ket{\psi}_K \otimes \ket{\phi}_{K'}$}.
\end{lemma}

\begin{proof}
The average of $\rho^{\otimes 2}$ is given by the sum over the Clifford group. Since the Clifford group forms a 2-design, using the twirling of states relation \cref{eqn Twirling of states} this sum can be replaced with integration over the full Haar measure,
\begin{subequations}
\begin{align}
\bigg\langle \ket{\psi}\bra{\psi}_{AB}\otimes &\ket{\psi}\bra{\psi}_{A'B'}\bigg\rangle \notag \\
&=\frac{1}{K}\sum_{k=1}^K \left(U_k\otimes U_k \right) \left( \ket{\tilde{\psi}}\bra{\tilde{\psi}}_{AB} \otimes \ket{\tilde{\psi}}\bra{\tilde{\psi}}_{A'B'}\right) \left(U_k^\dagger \otimes U_k^\dagger \right) \\
& =\int_{\mathcal{U}(d)} dU \left(U\otimes U \right) \left( \ket{\tilde{\psi}}\bra{\tilde{\psi}}_{AB} \otimes \ket{\tilde{\psi}}\bra{\tilde{\psi}}_{A'B'}\right) \left(U^\dagger \otimes U^\dagger \right).
\end{align}
\end{subequations}
Therefore the above average is invariant under conjugation by $(V\otimes V)$ for any unitary $V$,
\begin{equation}
(V\otimes V) \big\langle\ket{\psi}\bra{\psi}_{AB}\otimes \ket{\psi}\bra{\psi}_{A'B'} \big\rangle  (V^\dagger \otimes V^\dagger)= \big\langle\ket{\psi}\bra{\psi}_{AB}\otimes \ket{\psi}\bra{\psi}_{A'B'} \big\rangle.
\end{equation}
The operator $\big\langle\rho^{\otimes 2}\big\rangle$ only has support on the symmetric subspace of $\mathcal{H}_{AB}\otimes \mathcal{H}_{A'B'}$ and commutes with every element of $U\otimes U$.
\cref{thm sym subspaces} implies that $U\otimes U$ is irreducible on the symmetric subspace.
Since $\big\langle\rho^{\otimes 2}\big\rangle$ commutes with every element of this irrep, it must be proportional to the identity operator on that subspace by Shur's Lemma.
The identity operator on the symmetric subspace is the projector onto that subspace, which can be expressed as $\Pi^\text{sym}_{ABA'B'} = \frac{1}{2}(\mathbb{I}_{ABA'B'}+\mathcal{F}_{ABA'B'})$.
It then follows from normalisation that
\begin{subequations}
\begin{align}
\big\langle \ket{\psi}\bra{\psi}_{AB}\otimes \ket{\psi}\bra{\psi}_{A'B'}\big\rangle &= \frac{1}{\text{\tiny dim of sym subspace of } {\scriptscriptstyle ABA'B'}} \frac{1}{2}(\mathbb{I}_{ABA'B'}+\mathcal{F}_{ABA'B'})\\
& = \frac{2}{d_{AB}(d_{AB}+1)} \frac{1}{2}(\mathbb{I}_{ABA'B'}+\mathcal{F}_{ABA'B'}).
\end{align}
\end{subequations}

\end{proof}

We arrive at a value for the expectation value of the purity of a random stabiliser state, which will later appear in our application of Levy's lemma.

\begin{lemma}[Average of the purity]
\label{lm Average of the purity}
Given a stabilizer state $\ket{\psi}_{AB}$ on $t$ qudits of prime dimension $p$, bipartitioned into subsets $A$ and $B$ of $t_A$ and $t_B$ qudits respectively where $t_A \leq t_B$. The average purity of the reduced density matrix is given by,
\begin{equation}
\big\langle \trace_A (\rho_A^2) \big\rangle = \frac{d_A+d_B}{d_{AB}+1}
\end{equation}
where $d_A=p^{t_A}$, $d_B=p^{t_B}$, $d_AB=p^{t_A+t_B}$ are the dimensions of the respective (sub)spaces.
\end{lemma}

\begin{proof}
Considering a second copy of the total Hilbert space: $\mathcal{H}_{T'} = \mathcal{H}_{A'} \otimes \mathcal{H}_{B'}$ one can apply the swap trick,
\begin{subequations}
\begin{align}
\trace_A (\rho_A^2) & = \trace_{AA'}\left[(\rho_A\otimes \rho_{A'}) \mathcal{F}_{AA'} \right]\\
& = \trace_{TT'}\big[(\ket{\psi}\bra{\psi}_{T}\otimes \ket{\psi}\bra{\psi}_{T'}) (\mathcal{F}_{AA'}\otimes \mathbb{I}_{BB'})\big].
\end{align}
\end{subequations}
The average of this function can be calculated by considering the average over the random states,
\begin{equation}
 \big\langle \trace_A (\rho_A^2) \big\rangle= \trace_{TT'}\left[ \big\langle\ket{\psi}\bra{\psi}_{T}\otimes \ket{\psi}\bra{\psi}_{T'}\big\rangle (\mathcal{F}_{AA'}\otimes \mathbb{I}_{BB'})\right].
\end{equation}
Substituting into the above the average of $\rho^{\otimes 2}$ from \cref{lm Average of rho}, using $d_T$ to denote the dimension of the total Hilbert space $\mathcal{H}_T$,
\begin{equation}
\big\langle \trace_A (\rho_A^2) \big\rangle = \trace_{TT'}\left[ \frac{\left(\mathbb{I}_{TT'}+\mathcal{F}_{TT'}\right)}{d_T(d_T +1)} (\mathcal{F}_{AA'}\otimes \mathbb{I}_{BB'})\right].
\end{equation}
Writing $\mathcal{F}_{TT'} = \mathcal{F}_{AA'} \otimes \mathcal{F}_{BB'}$ and noting that $\mathcal{F}_{KK'}^2 = \mathbb{I}_{KK'}$:
\begin{subequations}
\begin{align}
\big\langle \trace_A (\rho_A^2) \big\rangle &= \trace_{TT'}\left[\frac{\left(\mathbb{I}_{TT'}+\mathcal{F}_{AA'}\otimes\mathcal{F}_{BB'} \right)}{d_T(d_T +1)} (\mathcal{F}_{AA'}\otimes \mathbb{I}_{BB'})\right]\\
& = \trace_{TT'}\left[ \frac{\mathbb{I}_{TT'}(\mathcal{F}_{AA'}\otimes \mathbb{I}_{BB'})}{d_T(d_T +1)} \right] + \trace_{TT'}\left[ \frac{\mathbb{I}_{TT'}(\mathbb{I}_{AA'}\otimes \mathcal{F}_{BB'})}{d_T(d_T +1)} \right]\\
& \begin{multlined} =\trace_{TT'}\left[ \left(\frac{\mathbb{I}_{T}}{d_T} \otimes \frac{\mathbb{I}_{T'}}{d_{T}+1} \right) (\mathcal{F}_{AA'}\otimes \mathbb{I}_{BB'})\right]  \\
\hspace{60pt}+ \trace_{TT'}\left[ \left(\frac{\mathbb{I}_{T}}{d_T} \otimes \frac{\mathbb{I}_{T'}}{d_{T}+1} \right) (\mathbb{I}_{AA'}\otimes \mathcal{F}_{BB'})\right]
\end{multlined}\\
& \begin{multlined} =\frac{d_T}{d_T+1} \left(\trace_{TT'}\left[ \left(\frac{\mathbb{I}_{T}}{d_T} \otimes \frac{\mathbb{I}_{T'}}{d_{T}} \right) (\mathcal{F}_{AA'}\otimes \mathbb{I}_{BB'})\right]  \right. \\
\hspace{68pt} \left. +\trace_{TT'}\left[ \left(\frac{\mathbb{I}_{T}}{d_T} \otimes \frac{\mathbb{I}_{T'}}{d_{T}} \right) (\mathbb{I}_{AA'}\otimes \mathcal{F}_{BB'})\right]\right)
\end{multlined}\label{eqnbeforeswap}\\
& = \frac{d_T}{d_T+1} \left( \trace_A \left[\left(\frac{\mathbb{I}_A}{d_A} \right)^2 \right] + \trace_B \left[\left(\frac{\mathbb{I}_B}{d_B} \right)^2 \right] \right)\label{eqnafterswap}\\
& = \frac{d_A + d_B}{d_Ad_B+1}.
\end{align}
\end{subequations}
Where between lines \cref{eqnbeforeswap} and \cref{eqnafterswap} we have used the swap trick in reverse to obtain the result.
\end{proof}

\subsubsection{Quantised entropy of stabilizer states}
\label{Quantised entropy of stabilizer states}

Applying Levy's lemma to the purity of Haar random states, along with the bounds on moments from \cref{lm moment bound}, is already sufficient to conclude that high dimensional random stabilizer tensors are \textit{close} to perfect with high probability \footnote{$S(\rho_A)$ is always lower bounded by $S_2(\rho_A)$}.
Some properties of perfect tensors will follow approximately from this statement. For example, the tensor is an approximate isometry across any bipartition where the departure can be suppressed by scaling the bond dimension. One could individually investigate the behaviours of this approximate perfect tensor to see if the deviation can be suppressed in all relevant cases. Instead we look to exploit the algebraic structure of the stabilizer group to demonstrate a strengthened result: high dimensional random stabilizer tensors are \textit{exactly} perfect still with high probability.

The following theorem builds on results (see e.g.~\cite{stabilizer}) showing that the entropy of a bipartite stabilizer state is quantised.

\begin{theorem}[Quantisation of entropy]
\label{thm Quantisation of entropy}
Given a stabilizer state $\ket{\psi}_{AB}$ on $t$ qudits of prime dimension $p$, bipartitioned into subsets $A$ and $B$ of $t_A$ and $t_B$ qudits respectively where $t_A \leq t_B$. The reduced density matrix, $\rho_A$, has a flat spectrum and its entropy, $S(\rho_A)$, is quantised in units of $\log p$.
\end{theorem}

\begin{proof}
Let $S$ be the stabilizer of $\ket{\psi}_{AB}$ and $S_A$ be the subgroup of $S$ consisting of all stabilizer elements that act with identity on $B$. Let $|S_A|$ denote the number of elements in the subgroup $S_A$. Since $g\rho = \rho$ for all $g\in S$ the density matrix and reduced density matrix can be written as (see e.g.~\cite{stabilizer} equations 4,5.)
\begin{equation}
\rho = \ket{\psi}\bra{\psi} = \frac{1}{p^t}\sum_{g \in S} g,
\end{equation}
\begin{equation}
\rho_A  = \trace_B\left(\ket{\psi}\bra{\psi}\right)=\frac{1}{p^{t_A}} \sum_{g\in S_A} g.
\end{equation}
Now we can show that $\rho_A$ is proportional to a projector,
\begin{subequations}
\begin{align}
\rho_A^2 & = \left( \frac{1}{p^{t_A}}\sum_{g\in S_A} g \right)^2\\
&= \frac{1}{p^{2t_A}} \sum_{g,g'\in S_A} g g' \\
& =\frac{1}{p^{2t_A}} \sum_{g,g''\in S_A} g'' \tag*{ by the Rearrangement theorem, $\{hg|g\in G\}=G$}\\
& = \frac{|S_A|}{p^{2t_A}} \sum_{g'' \in S_A} g'' \\
& = \frac{|S_A|}{p^{t_A}}\rho_A.
\end{align}
\end{subequations}
Therefore since $\rho_A^2 \propto \rho_A$ it has a flat spectrum with the eigenvalues $\lambda = \frac{|S_A|}{p^{t_A}}$ or $\lambda=0$. The size of the subgroup $S_A$ is given by $|S_A|=p^x$ where $x\leq t_A$ is the number of stabilizer generators, giving $p^{t_A-x}$ non-zero eigenvalues $=p^{x-t_A}$ (since $\trace\rho_A=1$). Consider the von Neumann entropy of $\rho_A$ (although since stabilizer states have flat entanglement spectrum all R\'enyi entropies are equal to the von Neumann entropy),
\begin{subequations}
\begin{align}
S(\rho_A) &= - \trace \left(  \rho_A \log \rho_A \right) \\
& = -\sum_j \lambda_j \log \lambda_j \\
& = -\sum_{j=1}^{p^{t_A-x}} \frac{|S_A|}{p^{t_A}} \log \frac{|S_A|}{p^{t_A}}\\
& = - \sum_{j=1}^{p^{t_A-x}} p^{x-t_A} \log (p^{x-t_A})\\
& = (t_A-x) \log p.
\end{align}
\end{subequations}
$t_A$ and $x$ are both integers hence $S(\rho_A)$ is quantised in units of $\log p$.
\end{proof}

\subsubsection{Random stabilizer tensors are perfect tensors with high probability}
\label{Random stabilizer tensors are perfect tensors with high probability}

Our first key result is that a random stabilizer tensor is exactly perfect with probability that can be pushed arbitrarily close to 1 by scaling the bond dimension, $p$. Formally,

\begin{theorem}[Random stabilizer tensors are perfect]
\label{thm Random stabilizer tensors are perfect}
Let the tensor $T$, with $t$ legs, describe a stabilizer state $\ket{\psi}$ chosen uniformly at random where each leg corresponds to a prime $p$-dimensional qudit. The tensor $T$ is perfect in the sense of \cref{defn Perfect tensors} with probability
\begin{equation} \label{high_prob_eqn}
P \geq  \max \left\{ 0, 1 - \frac{1}{2p^b}{t\choose \floor{t/2}}\right\}
\end{equation}
in the limit where $p$ is large, Where $0< b \leq 1$.
 \end{theorem}

\begin{proof}
A sufficient condition for a tensor to be perfect is that the reduced density matrix of every subset of legs $|A|=\floor*{t/2}$ is maximally mixed, using condition (ii) from \cref{defn AME} of an AME state. We first use concentration results to find the probability of being maximally entangled across a given bipartition with $d_A=p^{\floor*{t/2}}$ and $d_B=p^{\ceil*{t/2}}$.

Applying Levy's lemma to the purity (using the bound on the Lipschitz constant found in \cref{Lipschitz constant}, \cref{lm purity Lipschitz}) gives a bound on the probability tails for Haar random states, $\ket{\phi}_{AB}$,
\begin{subequations}
\begin{align}
\text{Prob}_H \left( |\trace(\sigma_A^2) - \mu| \geq \epsilon \right) &\leq 2 \exp \left(\frac{-4C_1 p^t \epsilon^2}{\eta^2} \right)\\
& \leq 2 \exp \left(\frac{-4C_1 p^t \epsilon^2}{4} \right)\\
& \leq 2 \exp \left(-C_1 p^t \epsilon^2 \right),
\end{align}
\end{subequations}
where $\sigma_A = \trace_B \left(\ket{\phi}\bra{\phi} \right)$ and $\mu$ is the mean of purity.
Since purity is a degree-2 polynomial, and stabililizer states form a 2-design, the expectation value over random stabilizer states and Haar random states coincide (see \cref{Expectation of purity}).

This can be combined with \cref{lm moment bound} where $m=1$, $C_2=2$ and $a = C_1p^t$ to give
\begin{subequations}
\begin{align}
\langle|\trace(\sigma_A^2)-\mu|\rangle_H&\leq C_2 \left( \frac{m}{2a}\right)^{m/2}\\
&= 2\left(\frac{1}{2C_1p^t} \right)^{1/2}.\label{eqn bound on dev avg}
\end{align}
\end{subequations}

Starting with Markov's inequality, \cref{eqn bound on dev avg} is used to upper bound the probability that the purity of stabilizer states is higher than the mean:
\begin{subequations}
\begin{align}
\text{Prob}_{k=2} \left(\trace(\rho_A^2)\geq \mu + \delta \right)&\leq \frac{\langle\trace(\rho_A^2)\rangle_{k=2}}{\mu+ \delta}\\
&=\frac{\langle\trace(\sigma_A^2)\rangle_{H}}{\delta + \mu}\label{eqn 2 des to haar step}\\
&=\frac{\langle\trace(\sigma_A^2)-\mu\rangle_{H} + \mu}{\delta + \mu}\\
&\leq \frac{\langle|\trace(\sigma_A^2)-\mu|\rangle_{H} + \mu}{\delta + \mu}\\
&\leq \frac{2 \sqrt{\frac{1}{2C_1p^t} }+ \mu}{\delta + \mu}\label{eqn prob1},
\end{align}
\end{subequations}
where $\rho_A = \trace_B \left(\ket{\psi}\bra{\psi} \right)$.
We have considered the $\trace (\rho_A^2)>\mu$ case since it is the lower tail of $S(\rho_A)$ we are interested in bounding.
In \cref{eqn 2 des to haar step} we have replaced the expectation over random stabilizer states $\rho_A$ with the expectation over Haar random states $\sigma_A$, since $\trace(\rho_A^2)$ is a degree-2 polynomial and stabilizer states form a 2-design.

\cref{thm Quantisation of entropy} demonstrated that the eigenspectrum is flat and therefore all R\'enyi entropies are equal.
Therefore we can relate the above probability statement to the von Neumann entropy via the R\'enyi-2 entropy using $S(\rho_A) = S_2(\rho_A) = -\log(\trace(\rho_A^2))$,
\begin{subequations}
\begin{align}
\text{Prob}_{k=2} \left(\trace(\rho_A^2) \geq \mu + \delta \right) &=  \text{Prob}_{k=2} \left(\log\trace(\rho_A^2) \geq \log(\mu + \delta) \right)\\
& = \text{Prob}_{k=2} \left(-\log\trace(\rho_A^2) \leq -\log(\mu + \delta) \right)\\
& = \text{Prob}_{k=2} \left(S(\rho_A) \leq -\log(\mu + \delta) \right).
\end{align}
\end{subequations}
Using \cref{lm Average of the purity} to substitute the average purity $\left( \mu = \frac{d_A + d_B}{d_Ad_B+1}\right)$ further manipulation gives this probability in terms of deviation from the maximum entropy:
\begin{align}
\text{Prob}&_{k=2} \left(\trace(\rho_A^2) \geq \mu +  \delta \right) \notag\\
&= \text{Prob}_{k=2} \left( S(\rho_A) \leq \log d_A - \underbrace{\log \left(\frac{d_A+d_B}{d_B + 1/d_A} + \delta d_A \right)}_{\Delta} \right) \label{eqn prob2}.
\end{align}
Combining \cref{eqn prob1} and \cref{eqn prob2} gives
\begin{equation}
\text{Prob}_{k=2} \left( S(\rho_A) \leq \log d_A - \Delta \right) \leq \frac{2 \sqrt{\frac{1}{2C_1p^t} }+ \mu}{\delta + \mu}. \label{eqn apply quant cond}
\end{equation}
In the limit of large $p$ this bound scales as:
\begin{equation}
\frac{2 \sqrt{\frac{1}{2C_1p^t} }+ \mu}{\delta + \mu}=O\left( \frac{p^{-t/2}+ p^{-\floor*{t/2}}}{\delta + p^{-\floor*{t/2}}}\right) =O\left( \frac{1}{\delta p^{\floor*{t/2}}}\right). \label{scaling_rhs}
\end{equation}

We know that the entropy is quantised in units of $\log(p)$ (see \cref{thm Quantisation of entropy}). Therefore if the deviation from the mean, $\Delta$, is less than $\log(p)$, then \cref{eqn apply quant cond} describes the probability of the entropy being less than its maximum value. The quantisation is less than the quantisation unit if:
\begin{equation}
\Delta \leq\log p \label{eqn quant cond}
\end{equation}
\begin{equation}
\frac{d_A+d_B}{d_B + 1/d_A} + \delta d_A  \leq  p
\end{equation}
\begin{equation}
\frac{p^{\floor*{t/2}}+p^{\ceil*{t/2}}}{p^{\ceil*{t/2}} + p^{-\floor*{t/2}}} + \delta p^{\floor*{t/2}}  \leq  p.
\end{equation}
For any $t$, the following marginally stronger condition on $\delta$ will also ensure that \cref{eqn quant cond} is satisfied
\begin{subequations}
\begin{align}
2 + \delta p^{\floor*{t/2}} &\leq p\\
\delta & \leq p^{-\floor*{t/2}}(p-2),
\end{align}
\end{subequations}
which is satisfied by $\delta = O(p^{b-\floor*{t/2}})$ with $0\leq b < 1$. Combining this and the bound from \cref{scaling_rhs}, we find that in the limit of large $p$, \cref{eqn apply quant cond} becomes,
\begin{subequations}
\begin{align}
\text{Prob}_{k=2} \left( S(\rho_A) \neq \log d_A \right) &\leq \frac{1}{\delta p^{\floor*{t/2}}}\\
&\leq \frac{1}{ p^{b}}
\end{align}
\end{subequations}
and the state is maximally entangled across a \textit{given} equal bipartition with probability,
\begin{equation}
P' \geq 1 -\frac{1}{p^{b}}.
\end{equation}

Maximal entropy across \textit{every} bipartition is required for the state to be AME and hence the tensor to be perfect. The number of distinct ways to equally bipartition a tensor of $t$ legs is ${t\choose \floor*{t/2}}/2$. Making no assumption about the dependence of the events $A_i$, Fr\'echet inequalities~\cite{Frechet1, Frechet2} bound the conjunction probability of $N$ events by:
\begin{multline}
\max \left\{0, P(A_1)+ P(A_2) +... +P(A_N) -(N-1)\right\} \leq P(A_1 \cap A_2 \cap ... \cap A_N) \\ \leq \min \left\{P(A_1), P(A_2),...P(A_N) \right\}.
\end{multline}
Hence the joint probability of all ${t\choose \floor*{t/2}}/2$ bipartitions satisfying $S(\rho_A)=\log d_A$ is lower bounded by
\begin{subequations}
\begin{align}
P \geq & \max \left\{ 0, {t\choose \floor{t/2}}\frac{P'}{2} -  \left[ {t\choose \floor{t/2}}/2 -1 \right] \right\} \\
\geq & \max \left\{ 0, 1 - \frac{1}{2p^b}{t\choose \floor{t/2}}\right\}.
\end{align}
\end{subequations}
\end{proof}

Consequently by making $p$ arbitrarily large, $P'$ and subsequently $P$ can be pushed close to 1 independent of $t$. Therefore by scaling the bond dimension we can ensure that, with high probability, the random stabilizer tensor is exactly a perfect tensor describing a qudit stabilizer AME state.

It follows from this result that an individual random stabilizer tensor inherits all the properties of perfect stabilizer tensors with probability that can be made arbitrarily close to 1 by increasing the bond dimension $p$. That is, they are an isometry across any bipartition and describe a family of stabilizer MDS codes, where there exists a consistent choice of basis that preserves the Pauli rank of the operator.

\subsection{Entanglement structure}
\label{Entanglement structure}

We now investigate the entanglement structure of our proposed holographic code, since the principal motivation for choosing random tensors was to demonstrate a construction that exactly obeys Ryu-Takayanagi while simultaneously mapping between local Hamiltonians. In this section we show that with high probability, when the bond dimension is large, our HQECC construction obeys the Ryu-Takayanagi entropy formula exactly for all boundary regions. We demonstrate this rigorously for arbitrary unentangled bulk states. We lean heavily on the work of~\cite{Random}, where by mapping to a classical spin system they were able to say something about the entropies in a random network.

First using their lower bound on the average entanglement entropy we make an exact statement of Ryu-Takayanagi for stabilizer product bulk states, while the general product state case is only approximate. Reusing their mapping Hayden et al. also present near complementary recovery via an entropic argument. We follow their proof structure but taking care and refining our network set-up to ensure any local operator acting on a bulk index can be recovered on either subregion for all bipartitions of the boundary. Finally we use this complementary recovery result to elevate the exact statement of Ryu-Takayanagi to apply to arbitrary product bulk states.

\subsubsection{Approximate Ryu-Takayanagi}

Hayden et al. investigated properties of general random tensor networks in~\cite{Random}, particularly their entanglement structure. Their methods of interpreting functions of the tensor network as partition functions of the classical Ising model led to Ryu-Takayanagi minimal surfaces manifesting as domain walls between spin regions. This method will be a key technique in demonstrating the improved entanglement structure of our modified code, so we encapsulate their argument in the following result. The logic assumes nothing about the geometry of the network graph and uses only second order moments so can be applied to our construction. For technical details of the proof we refer to~\cite{Random}.

\begin{lemma}[Mapping to an Ising partition function; discussion in~\cite{Random} section 2]
\label{lm Mapping to an Ising partition function}
Let $\rho_\mathcal{B}$ be the density operator of a boundary state obtained by mapping a bulk state through a random HQECC, where tensors have a bond dimension $p$. Let $A$ be any subregion of the boundary.

Introduce a spin interpretation of the network, as follows:
\begin{itemize}
\item To every vertex $x$ in the network graph -- whether there be a tensor there or a dangling bulk/boundary index -- assign an Ising variable, $s_x\in \pm1$. Let $s_x=+1$ correspond to acting at that index with the identity operator and $s_x=-1$ to acting instead with the swap operator, $\mathcal{F}_x$, defined on two copies of the original system $(\rho_\mathcal{B} \otimes \rho_\mathcal{B})$.
\item To every boundary vertex also assign a pinning field parameter, $h_x \in \pm 1$, where $h_x = -1$ if $x \in A$ and $h_x=+1$ otherwise.
\end{itemize}
The following function, inspired by the form of $S_2(\rho_A)$, can be manipulated into the form of a partition function of the above classical spin system:
\begin{equation}
Z := \langle \trace \left[ (\rho_\mathcal{B} \otimes \rho_\mathcal{B}) \mathcal{F}_A \right] \rangle = \sum_{\{s_x \}} e^{-\mathcal{A}[\{s_x \}]},
\end{equation}
where the sum is over all possible configurations of the Ising parameters, $\{s_x \}$, and the Ising action is given by
\begin{equation}
\mathcal{A}[\{s_x \}] = S_2(\{s_x=-1 \};\rho_b) - \sum_{\langle xy \rangle} \frac{1}{2}(s_x s_y -1) \ln p - \sum_{x\in \mathcal{B}} \frac{1}{2}(h_x s_x -1) \ln p + \textup{const}.
\end{equation}
$S_2(\{s_x=-1 \};\rho_b)$ denotes the R\'enyi-2 entropy of the density operator of the bulk degrees of freedom restricted to the domain where $s_x=-1$.
\end{lemma}

The above lemma has interesting consequences when considering an unentangled bulk state ($S_2(\{s_x=-1 \};\rho_b)=0$) since,
\begin{subequations}
\begin{align}
Z &= \sum_i e^{-\beta E_i}\\
&= \sum_{\{s_x \}} \exp \left[- \ln p \left( \frac{1}{2}\sum_{\langle xy \rangle} (s_x s_y -1) + \frac{1}{2}\sum_{x\in\mathcal{B}}(h_x s_x-1)\right) \right],
\end{align}
\end{subequations}
where we identify $\beta = \ln p$ and $E_{\{ s_x\}} = \frac{1}{2}\sum_{\langle xy \rangle} (s_x s_y -1) + \frac{1}{2}\sum_{x\in\mathcal{B}}(h_x s_x-1)$.

In the large $p$ (low temperature) regime this partition function is dominated by the ground state energy of the system with low energy fluctuations. The ground state of an Ising model in $>1$ dimensions is given by the configuration of $\{s_x\}$ that minimises the domain wall `length'. Hence the ground state energy, $E_{GS}$, is given by the number of tensor `legs' crossed by the Ryu-Takayanagi surface of the boundary region $A$, $|\gamma_A|$. \cref{Low temperature Ising model corrections} demonstrates that in this regime $-\ln Z$ deviates from the ground state energy by manageable corrections,
\begin{subequations}
\begin{align}
-\ln Z &= \beta E_{GS} - \ln k - O \left( \frac{1}{p} \right)\\
&= \ln (p) |\gamma_A| - \ln k - O \left( \frac{1}{p} \right).
\end{align}
\end{subequations}
Here $k$ is the number of minimal geodesic surfaces.

The above object is close to the Ryu-Takayanagi entropy which we identify with $S_{RT}(\rho_A)=\ln (p) |\gamma_A|$ when there is no entanglement in the bulk. It remains to relate $-\ln Z = -\ln \langle \trace \left[ (\rho_\mathcal{B} \otimes \rho_\mathcal{B}) \mathcal{F}_A \right] \rangle$ back to the entanglement entropy of the boundary subregion. For an arbitrary state the entanglement entropy is lower bounded by the R\'enyi-2 entropy,
\begin{equation}
S_2(\rho_A) = -\ln \trace \left[\rho_A^2 \right] = -\ln \trace \left[ \left(\rho_\mathcal{B}\otimes \rho_\mathcal{B} \right)\mathcal{F}_A \right].
\end{equation}
The average R\'enyi-2 entropy can be expanded as $\langle S_2(\rho_A) \rangle  \approx -\ln \langle \trace \left[ (\rho_\mathcal{B} \otimes \rho_\mathcal{B}) \mathcal{F}_A \right] \rangle$ plus some fluctuations. Using only second order moments Hayden et al. consider these fluctuations, proving a lower bound on the average entanglement entropy. Their result is encapsulated in the following lemma; for the proof we refer to their work.

\begin{lemma}[Lower bound on the entanglement entropy; discussion in appendix F of~\cite{Random}]
\label{lm Lower bound on the entanglement entropy}
The average entanglement entropy of a general boundary region $A$ of a random stabilizer tensor network is lower bounded by
\begin{equation}
S_{RT}(\rho_A) - \left[\ln k +o(1)\right]\leq \langle S_2(\rho_A) \rangle_{\neq 0}\leq \langle S(\rho_A) \rangle_{\neq 0 },
\end{equation}
where $\langle \rangle_{\neq 0}$ is the average over all choices of network excluding the cases resulting in $\rho_\mathcal{B} = 0$ where the entropy is not well-defined, and $k$ is the number of minimal geodesic surfaces.
 \end{lemma}

Using the above results we come to our first statement about the entanglement entropy of general boundary subregions that can be applied to our construction.

\begin{lemma}[Approximate Ryu-Takayanagi]
\label{lm Approximate Ryu-Takayanagi}
For a random stabilizer tensor network with bond dimension $p$, let $S(\rho_A)$ be the entanglement entropy of $A$, any (disconnected or connected) subregion of the boundary. Given
\begin{enumerate}
\item an arbitrary tensor-product bulk input state:
\begin{equation}
\text{\emph{Prob}} \left[S_{RT}(\rho_A)-S(\rho_A) \geq a\cdot (\ln k +o(1)) \right] \leq \frac{1}{a}.
\end{equation}
Hence the entanglement entropy can be made to be $(a \cdot\ln k)$-close to the Ryu-Takayanagi entropy with probability $\left(1 - \frac{1}{a} \right)$ by scaling the bond dimension $p$.
\item a stabilizer tensor-product bulk input state, conditional on $\log p > a \cdot (\ln k +o(1))$:
\begin{equation}
\text{\emph{Prob}} \left[S_{RT}(\rho_A)=S(\rho_A)  \right] \geq 1- \frac{1}{a}.
\end{equation}
Therefore by scaling $p$, $a$ can be made arbitrarily large so that the probability of having \textbf{exactly} the Ryu-Takayanagi entropy is pushed to 1.
\end{enumerate}
$S_{RT}(\rho_A) = |\gamma_A|\ln (p)$ is the Ryu-Takayanagi entropy for an unentangled bulk state. $\rho_\mathcal{B}$ is the density operator of the resulting boundary state after the bulk state has been encoded by the tensor network such that $\rho_A = \trace_{\bar{A}}(\rho_\mathcal{B})$. $k$ is the number of minimal geodesics through the graph and $a>0$.
 \end{lemma}

\begin{proof}
The entanglement entropy of a boundary region for an arbitrary tensor network state is upper bounded by $S_{RT}(\rho_A)$~\cite{multipartite_entanglement}. \cref{lm Lower bound on the entanglement entropy} gives an almost matching lower bound on the conditional average of the entropy. Applying Markov's inequality to this bound we can conclude that:
\begin{equation} \label{eqn Markov entropy bound}
\text{Prob} \left[S_{RT}(\rho_A)-S(\rho_A) \geq a\cdot (\ln k +o(1)) \right] \leq \frac{1}{a}.
\end{equation}
Equivalently,
\begin{equation}
\text{Prob} \left[S_{RT}(\rho_A)-S(\rho_A) \geq \tilde{a} \right] \leq \frac{\ln k +o(1)}{\tilde{a}}.
\end{equation}
where $a,\tilde{a}>0$. This leads immediately to point 1.

The tensor network state over both bulk and boundary degrees of freedom is a stabilizer state since it is formed by contracting stabilizer tensors -- see appendix G of~\cite{Random} for a proof. If the bulk state is a stabilizer product state $\rho_\mathcal{B}$ is again a stabilizer state. Hence $S(\rho_A)$ is quantised in units of $\log p$ from \cref{thm Quantisation of entropy}. Rearranging \cref{eqn Markov entropy bound} gives,
\begin{equation}
\text{Prob} \left[S_{RT}(\rho_A)-S(\rho_A) < a\cdot (\ln k +o(1)) \right] \geq 1- \frac{1}{a}.
\end{equation}
Therefore if $\log p > a \cdot (\ln k +o(1))$, the only way the condition in the probability can be satisfied is if the entropies are equal, leading to point 2.
\end{proof}

This proof technique does not require that $A$ is a single connected boundary region since minimal domain walls ground states apply for disconnected boundary regions, depicted in \cref{fg disconnected regions}. This is a key advantage over the techniques used in the HaPPY paper~\cite{Happy} which do require a connected boundary region. Furthermore, the above result accounts for the possibility of multiple geodesics which may occur for our network's geometry.

\begin{figure}[tbp]
\centering
\includegraphics[trim={0cm 40cm 30cm 0cm},clip,scale=0.135]{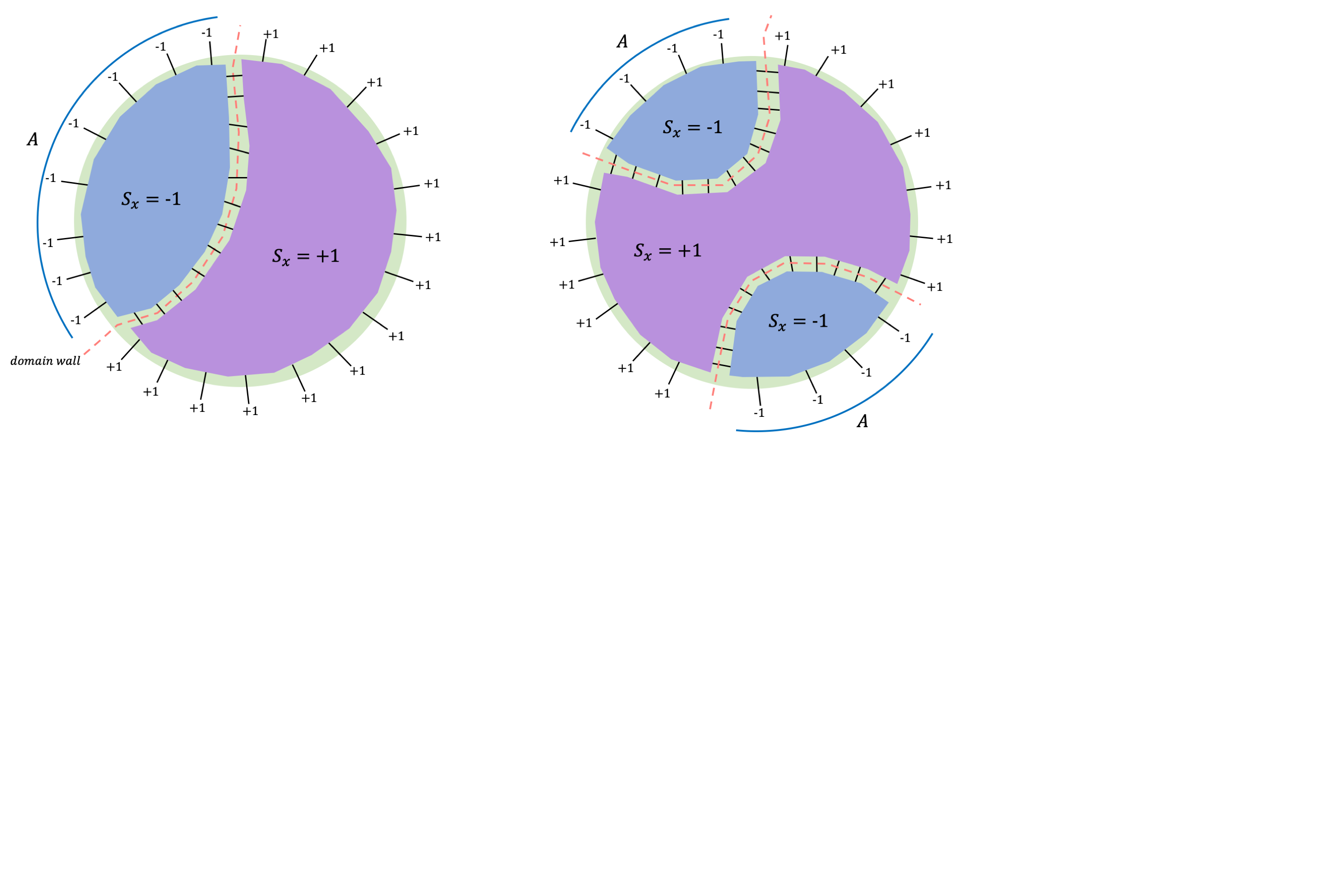}
\caption{Example of spin configurations used to calculate the entanglement entropy for connected boundary region (left) and disconnected region (right).}
 \label{fg disconnected regions}
\end{figure}

\subsubsection{Full complementary recovery}
\label{Full complementary recovery}

As discussed in \cref{Previous work}, complementary recovery is observed in AdS/CFT given any partition of the boundary into non-overlapping regions, conditional on the union of their entanglement wedges covering the entire bulk spacial slice.
In this section we will prove complementary recovery for our toy model using the Ising model technique introduced in the previous section.
For this quantitative proof, we assume firstly that the tensor network is complete, without any holes.
In~\cite{Happy, Tamara} deleting tensors from the network is associated with creating horizons and non-semiclassical states. Therefore a complete tensor network equates to an absence of black holes.
Secondly we only consider local bulk operators acting on a single bulk index and therefore our bulk is inherently unentangled.
These two assumptions ensure that bipartitions of the boundary have a common minimal area surface and hence we expect complementary recovery.
We refer to section 3.2 and 3.3 of~\cite{Random} for a qualitative discussion on the introduction of bulk entanglement and the resulting qualitative changes in the minimal surfaces.

Both the HaPPY code~\cite{Happy} and the extension by Kohler and Cubitt~\cite{Tamara} achieve approximate complementary recovery through the greedy entanglement wedge. However there exists certain `pathological' choices of boundary bipartitions where some local bulk operators cannot be recovered on either subregion. General random holographic codes in~\cite{Random} also realise approximate complementary recovery, where all but bulk indices in contact with the Ryu-Takayanagi surface can be recovered. We would expect full complementary recovery in a complete model of AdS/CFT and through careful choice of geometry we show that a HQECC based on random stabilizer tensors can achieve this.

\begin{lemma}[Complementary recovery]
\label{lm Complementary recovery}
For a tensor network comprised of random stabilizer tensors, with arbitrarily high probability we have full complementary recovery in the sense that any logical operator acting on any single bulk tensor index, $C$, can be recovered on either the boundary subregion ($A$) or its complement ($\bar{A}$) conditional on:
\begin{enumerate}
\item the hyperbolic bulk tessellation describing the network's geometry consisting of polytopes with an odd number of faces;
\item the dimension of a bulk dangling index, $D_b$, being less than that of internal connections in the network, $D$;
\item the bond dimension of the tensors being sufficiently large.
\end{enumerate}
\end{lemma}
\begin{proof}
Appendix B of~\cite{Random} shows that complementary recovery is equivalent to the following entropic equation:
\begin{equation}\label{compl_rec_entropies}
S(\rho_C) + S(\rho_{\bar{A}\bar{C}}) = S(\rho_{\bar{A}C\bar{C}}).
\end{equation}
We use a generalisation of the previous Ising model mapping summarised in \cref{Generalisation of the Ising mapping} to approximate these entropies, where again the errors can suppressed by scaling the tensor bond dimension. Each term is calculated by considering the energy penalties from the bulk and boundary pinning fields as well as domain walls of the ground state. The above equation is satisfied only if the spin domain walls described by the ground state corresponding to $S(\rho_{\bar{A}\bar{C}})$ and $S(\rho_{\bar{A}C\bar{C}})$ coincide. The only difference between the two set-ups is that the bulk pinning field on the vertex labelled $C$ switches sign. This section of the proof follows immediately from section 4 of~\cite{Random}.

Therefore for general bulk indices we do have complementary recovery. However there is a potential problem if $C$ describes a bulk tensor that is adjacent to the minimal domain wall associated with the boundary subregion $A$ i.e. the tensor has legs that cross the domain wall. In some cases the domain wall will move when the bulk pinning field is flipped, so that operators acting on that bulk tensor are not recoverable either on $A$ or $\bar{A}$.

This problematic scenario can be avoided by careful choice of tensor network. To see that the conditions listed above are sufficient to avoid these situations we consider the limiting case. Let $l$ be the number of tensor nearest-neighbours in the bulk, i.e. the number of faces of the bulk tessellation. Choose $C$ to be a bulk dangling index where the connected tensor is in the spin down domain when calculating $S(\rho_{\bar{A}\bar{C}})$. The bulk pinning field is spin-down, so at most $\floor*{l/2}$ of $C$'s legs can cross the domain wall, otherwise the lowest energy configuration would have $C$ in the spin-up domain. Considering the energy penalty trade off:
\begin{align}
&\text{energy penalty if $s_x = -1$ for $x\in C$: } \beta E_{-1} = \floor*{l/2} \log D \\
&\text{energy penalty if $s_x = +1$ for $x\in C$: } \beta E_{+1} = \ceil*{l/2} \log D  + \log D_b,
\end{align}
and
\begin{equation} \label{eqn energy inequality}
E_{+1} > E_{-1}.
\end{equation}
Then consider the case where the bulk pinning field for $S$ is flipped, $S(\rho_{\bar{A}C\bar{C}})$.
\begin{align}
&\text{energy penalty if $s_x = -1$ for $x\in C$: } \beta E_{-1} = \floor*{l/2} \log D + \log D_b,\\
&\text{energy penalty if $s_x = +1$ for $x\in C$: } \beta E_{+1} = \ceil*{l/2} \log D
\end{align}
Given $D_b<D$ and $l$ is odd ($\floor*{l/2}+1=\ceil*{l/2}$) the inequality in \cref{eqn energy inequality} is still true and the domain wall does not move.
\end{proof}

There exist space-filling tessellations of hyperbolic space in 2 and 3 dimensions.
In 2-dimensions one example is the tessellation composed of pentagons from~\cite{Happy}.
In 3-dimensions one example is based on a non-uniform Coxeter polytope with 7 faces and described in section 6.2 of~\cite{Tamara}.
We can meet the second condition without considering tensors with different dimensional indices by taking a tensor with $fs+1$ $p$-dimensional indices, where $f$ is the number of faces of each cell in the tessellation.
Then $s$ tensor indices will be contracted through each of the $f$ polytope faces, and one index will be the bulk degree of freedom.
So, in the Ising action $D_b=p$ and $D=p^s$ so that $D_b<D$.
We are free to choose $p$ arbitrarily large to satisfy the final condition of \cref{lm Complementary recovery} and ensure that the entropies in \cref{compl_rec_entropies} are arbitrarily well approximated by the free energy of the ground state of an appropriate Ising model.
Indeed, large bond dimension is also a requirement to achieve isometric tensors with high probability (\cref{thm Random stabilizer tensors are perfect}) so this condition is already required.

\subsubsection{Exact Ryu-Takayanagi}
\label{Exact Ryu-Takayanagi}

Exact Ryu-Takayanagi for any bulk state in tensor product form can be demonstrated by combining the two previous results.

\begin{lemma}[Exact Ryu-Takayanagi]
\label{lm Exact Ryu-Takayanagi}
We can construct a random stabilizer tensor network existing in 2 and 3-dimensional hyperbolic space such that the entanglement entropy of any (disconnected or connected) boundary subregion agrees exactly with the Ryu-Takayanagi entropy formula when there is no bulk entanglement. This occurs with probability:
\begin{equation}
\text{\emph{Prob}} \left[S_{RT}(\rho_A)=S(\rho_A)  \right] \geq 1- \frac{1}{a},
\end{equation}
conditional on $\log p > a \cdot (\ln k +o(1))$. All quantities carry their definitions from \cref{lm Approximate Ryu-Takayanagi}.
\end{lemma}
\begin{proof}

For the $\bigotimes \ket{0}$ bulk stabilizer state we have exact Ryu-Takayanagi for any boundary subregion from \cref{lm Approximate Ryu-Takayanagi}. From this stabilizer state we can get to an arbitrary tensor product state by applying local operators to bulk tensors in turn. If we have full complementary recovery as described in \cref{lm Complementary recovery}, the action of such a logical operator on the boundary cannot change the entropy of the two boundary subregions.

This can be seen since the application of an operator $\left(\mathbb{I}_A\otimes U_B \right)$ to $\rho$, commutes with taking the partial trace of $\rho$ to give $\rho_A$ or $\rho_B$:
\begin{align}
\trace_A \left[ \left(\mathbb{I}_A\otimes U_B \right) \rho \left(\mathbb{I}_A\otimes U_B \right)^\dagger \right] &= U_B \rho_B U_B^\dagger,\\
\trace_B \left[ \left(\mathbb{I}_A\otimes U_B \right) \rho \left(\mathbb{I}_A\otimes U_B \right)^\dagger \right] &= \rho_A.
\end{align}
Moreover, entropy is invariant under a unitary change of basis, $S(U_B \rho_B U_B^\dagger) = S(\rho_B)$. Therefore the von Neumann entropies of $\rho_A$ and $\rho_B$ are unchanged after applying the logical operator $\left(\mathbb{I}_A\otimes U_B \right)$ to the total state.

There exist tensor network constructions that meet the conditions of full complementary recovery from \cref{lm Complementary recovery} which we described in \cref{Full complementary recovery}. Hence using together exact Ryu-Takayanagi result for a stabilizer bulk state and complementary recovery gives exact Ryu-Takayanagi for arbitrary product bulk states.
\end{proof}

The above theorem implies that the entanglement structure expected from AdS/CFT is achieved in our construction with high probability since by scaling the bond dimension, $a$ can be made arbitrarily large so that the probability of having exactly the Ryu-Takayanagi entropy can be made arbitrarily close to~1. Therefore the entanglement structure of a network comprised of random stabilizer tensors is closer to AdS/CFT than one built from perfect tensors where Ryu-Takayanagi does not generally apply.

\cite{Random} also explores the entanglement structure for entangled bulk states in Haar random tensor networks, examining qualitatively how introducing entanglement in the bulk leads to displacement of the minimal surface and increased entropy of the boundary region. The minimal surface will never enter bulk regions of sufficiently high entanglement leading to discontinuous jumps as the boundary region varies. This transition is speculatively linked to the Hawking-Page transition where upon increasing temperature a black hole emerges from the perturbed AdS bulk geometry~\cite{Hawking}. The same techniques are used to study boundary two-point correlation functions which decay as a power law when the bulk has hyperbolic geometry, defining the spectrum of scaling dimension. They find a separation in scaling dimensions that is expected in AdS/CFT where there is a known scaling dimension gap~\cite{spectrum1, spectrum2}. This analysis could be recycled to further investigate the entanglement structure of stabilizer random tensors and we expect these proof techniques to be more fruitful than those of the HaPPY paper.

\subsection{HQECC between local Hamiltonians with random stabilizer tensors}
\label{HQECC between local Hamiltonians with random stabilizer tensors}

The HQECC toy model of AdS/CFT described in~\cite{Tamara} comprised of a 3d tessellation of perfect stabilizer tensors of prime-powered dimension. As briefly discussed in \cref{Previous work} their model reflected several qualities of the holographic duality, most notably the mapping between models.
In~\cite{kohler2020translationallyinvariant} this was extended to a holographic mapping between local Hamiltonians in a 2-d bulk to local Hamiltonians on a 1-d boundary.
In previous sections we have demonstrated the desirable entanglement and error correcting properties of a holographic code where we inherit the geometry of~\cite{Tamara}'s construction but replace each perfect tensor with a random stabilizer tensor. Furthermore, \cref{thm Random stabilizer tensors are perfect} stated that individually a random stabilizer tensor is highly likely to be a perfect stabilizer tensor so all the successes of~\cite{Tamara}'s construction can also be retained.
 Therefore our construction will, with high probability, describe a bulk to boundary encoding that exhibits several key features of the AdS/CFT correspondence.

In~\cite{Tamara} the holographic mapping was defined between a local Hamiltonian in the bulk and a local Hamiltonian in the boundary.
However, it is desirable to relax the condition of locality in the bulk to allow for quasi-local bulk Hamiltonians which exhibit gravitational Wilson lines.
It was noted in~\cite{kohler2020translationallyinvariant} that the proof in~\cite{Tamara} can be extended to cover Hamiltonians which are not strictly local.
Here, we will define our holographic mapping between `quasi-local' Hamiltonians in the bulk, and local Hamiltonians in the boundary.
Where we define a quasi-local Hamiltonian as:

\begin{defn}[Quasi-local hyperbolic Hamiltonians]
  Let $\mathbb{H}^d$ denote\linebreak $d$-dimensional hyperbolic space, and let $B_r(x)\subset\mathbb{H}^d$ denote a ball of radius $r$ centred at $x$.
  Consider an arrangement of $n$ qudits in $ \mathbb{H}^d$ such that, for some fixed $r$, at most $k$ qudits and at least one qudit are contained within any $B_r(x)$.
  Let $Q$ denote the minimum radius ball $B_Q(0)$ containing all the qudits (which without loss of generality we can take to be centred at the origin).
  A quasi $k$-local Hamiltonian acting on these qudits can be written as:
  \begin{equation}
  H_\text{bulk} = \sum_Z h^{(Z)}
  \end{equation}
  where the sum is over the $n$ qudits, and each term can be written as:
  \begin{equation}
  h^{(Z)} = h_{\mathrm{local}}^{(Z)} h_{\mathrm{Wilson}}^{(Z)}
  \end{equation}
  where:
  \begin{itemize}
  \item $h_{\mathrm{local}}^{(Z)}$ is a term acting non-trivially on at most $k$ qudits which are contained within some $B_r(x)$
  \item $h_{\mathrm{Wilson}}^{(Z)}$ is a Pauli operator acting non trivially on at most $\BigO(L-x)$ qudits which form a line between $x$ and the boundary of $B_Q(0)$.
  \end{itemize}
\end{defn}

\begin{figure}[tbp]
\centering
\begin{tikzpicture}
\begin{scope}[on background layer]
         \node [inner sep=0pt] (1) at (0, 0) {\includegraphics[trim={4cm 3cm 5cm 0cm},clip,scale=0.45]{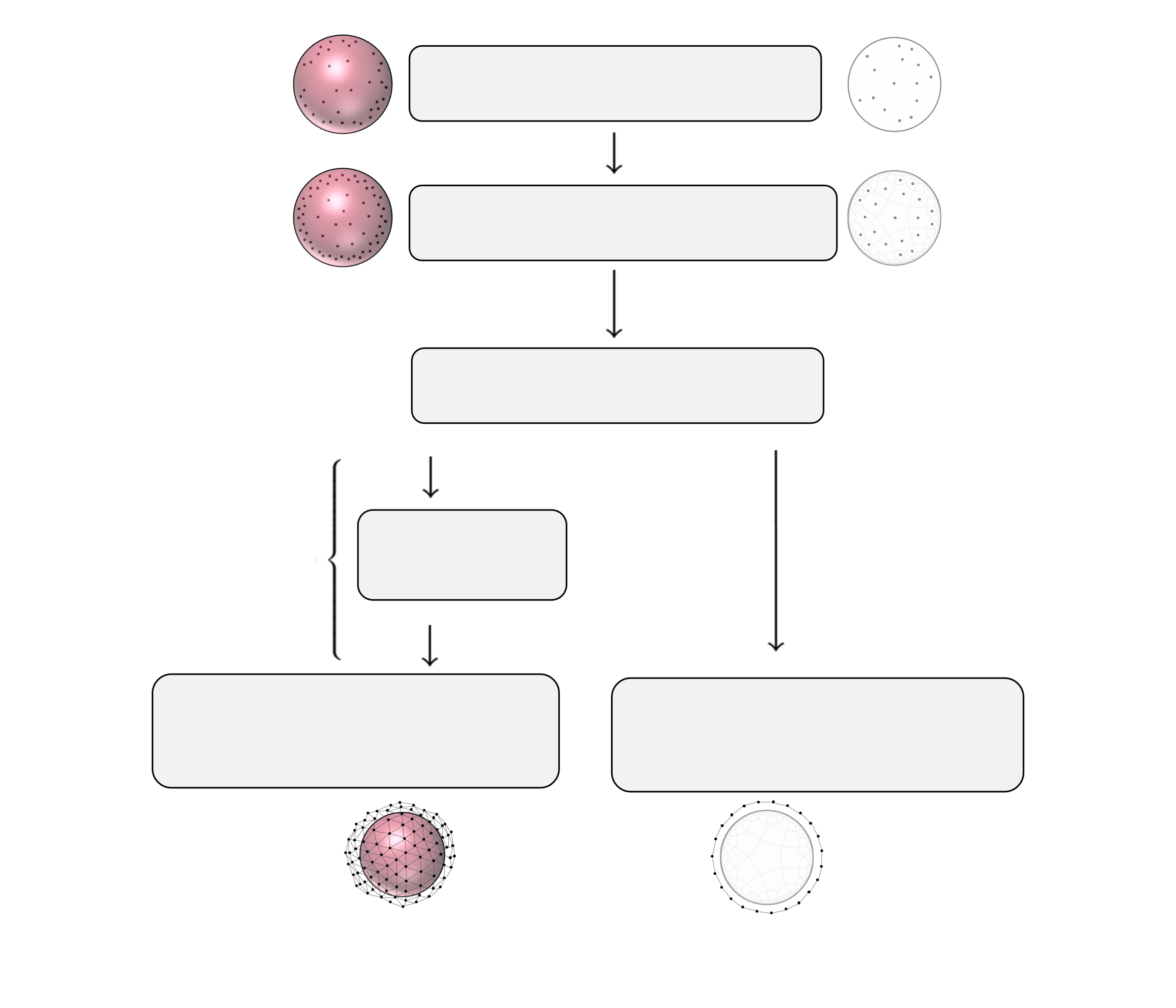}};
\end{scope}
		\node [align=left, anchor=west] (0) at (3.4, -0.7) {\scriptsize \cref{holography},};
		\node [align=left, anchor=west] (0) at (3.4, -1.05) {\scriptsize non-perturbative $(\Delta_L, \epsilon, \eta)-$};
		\node [align=left, anchor=west] (0) at (3.4, -1.4) {\scriptsize approximate simulation};
		\node [align=left, anchor=west] (0) at (3.4, -1.75) {\scriptsize $\Delta_L = \Omega(||H_\text{bulk}||)$};
		\node [align=left, anchor=west] (0) at (3.4, -2.1) {\scriptsize $\epsilon, \eta = 1/\text{poly}(\Delta_L)$};
		\node [align=left, anchor=west] (0) at (-7.3, -0.7) {\scriptsize \cref{thm Main result},};
		\node [align=left, anchor=west] (0) at (-7.3, -1.05) {\scriptsize perturbative $(\Delta_L, \epsilon, \eta)-$};
		\node [align=left, anchor=west] (0) at (-7.3, -1.4) {\scriptsize approximate simulation};
		\node [align=left, anchor=west] (0) at (-7.3, -1.75) {\scriptsize $\Delta_L = \Omega(||H_\text{bulk}||^6)$};
		\node [align=left, anchor=west] (0) at (-7.3, -2.1) {\scriptsize $\epsilon, \eta = 1/\text{poly}(\Delta_L)$};
		\node [align=left, anchor=west] (0) at (-3.7, 0.35) {\scriptsize Bulk dimension, $d\geq 3$};
		\node [align=left, anchor=west] (0) at (2, 0.35) {\scriptsize Bulk dimension, $d=2$};
		\node [align=left, anchor=west] (0) at (-3.4, -1.1) {\scriptsize Geometrically 2-local};
		\node [align=left, anchor=west] (0) at (-3.4, -1.45) {\scriptsize Hamiltonian acting on};
		\node [align=left, anchor=west] (0) at (-3.4, -1.8) {\scriptsize $O(n \log n)$ qudits};
		\node [align=left, anchor=west] (0) at (-6.7, -3.8) {\scriptsize Geometrically 2-local Hamiltonians acting on};
		\node [align=left, anchor=west] (0) at (-6.7, -4.15) {\scriptsize $O(n \poly(\log n))$ qudits in a tessellation of the};
		\node [align=left, anchor=west] (0) at (-6.7, -4.5) {\scriptsize boundary surface with $R' = O(\max(1, \frac{\ln(k) L}{r}$};
		\node [align=left, anchor=west] (0) at (-6.7, -4.85) {\scriptsize $+ \log \log n ))$};
		\node [align=left, anchor=west] (0) at (0.7, -3.8) {\scriptsize Geometrically 2-local Hamiltonians acting on};
		\node [align=left, anchor=west] (0) at (0.7, -4.15) {\scriptsize $O(n (\log n)^2)$ qudits in a tessellation of the};
		\node [align=left, anchor=west] (0) at (0.7, -4.5) {\scriptsize boundary surface with $R' = O(\max(1, \frac{\ln(k) L}{r}$};
		\node [align=left, anchor=west] (0) at (0.7, -4.85) {\scriptsize $+ \log \log n ))$};
		\node [align=left, anchor=west] (0) at (-2.5, 1.4) {\scriptsize Non-local boundary Hamiltonian acting on};
		\node [align=left, anchor=west] (0) at (-2.5, 1.05) {\scriptsize $O(n)$ boundary qudits};
		\node [align=left, anchor=west] (0) at (-2.7, 4.05) {\scriptsize (Quasi) $k$-local Hamiltonian on $O(n)$ bulk qudits};
		\node [align=left, anchor=west] (0) at (-2.7, 3.7) {\scriptsize embedded in a Coxeter polytope tessellation of $\mathbb{H}^d$};
		\node [align=left, anchor=west] (0) at (-2.6, 6.25) {\scriptsize (Quasi) $k$-local Hamiltonian on $n$ bulk qudits in};
		\node [align=left, anchor=west] (0) at (-2.6, 5.9) {\scriptsize $\mathbb{H}^d$ in a ball of radius $R=O(\max(1,\ln(k)L/r))$};
		\node [align=left, anchor=west] (0) at (1, 5) {\scriptsize Rescaling};
		\node [align=left, anchor=west] (0) at (1, 2.7) {\scriptsize High probability of a good random tensor};
		\node [align=left, anchor=west] (0) at (1, 2.35) {\scriptsize network when $p=O(n^q)$, \cref{thm Random stabilizer tensors are perfect}};
		\node [align=left, anchor=west] (0) at (-3.5, 2.52) {\scriptsize Perfect simulation below $\Delta_S$};
\end{tikzpicture}
\caption{Sequence of simulations used to create holographic toy models. Full details can be found in \cref{holography} for the 1d-2d case and \cref{thm Main result} for the higher dimensional case. The first two simulations are perfect and common between the $2d$ and $\geq3d$ bulk cases. These can be achieved with a random tensor network HQECC with probability that can be pushed arbitrarily close to 1 by scaling the tensor bond dimension, see equation \cref{eqn probability scaling}. Given this high probability we can efficiently obtain a good tensor network encoding by repeatedly construct random stabilizer networks and discarding those that do not describe a bulk to boundary isometry. The remaining simulations are ($\Delta_L, \epsilon, \eta$)-approximate, see~\cite{simulation} for exploration of how errors in physical observables scale with these error parameters. But by tuning Hamiltonian terms we are able to make these approximations arbitrarily accurate. The final result in both 2d and 3d is a dual local Hamiltonian living on a boundary surface that is qualitatively close in radius to the boundary surface of the bulk.}
 \label{fg simulation flow chart}
\end{figure}
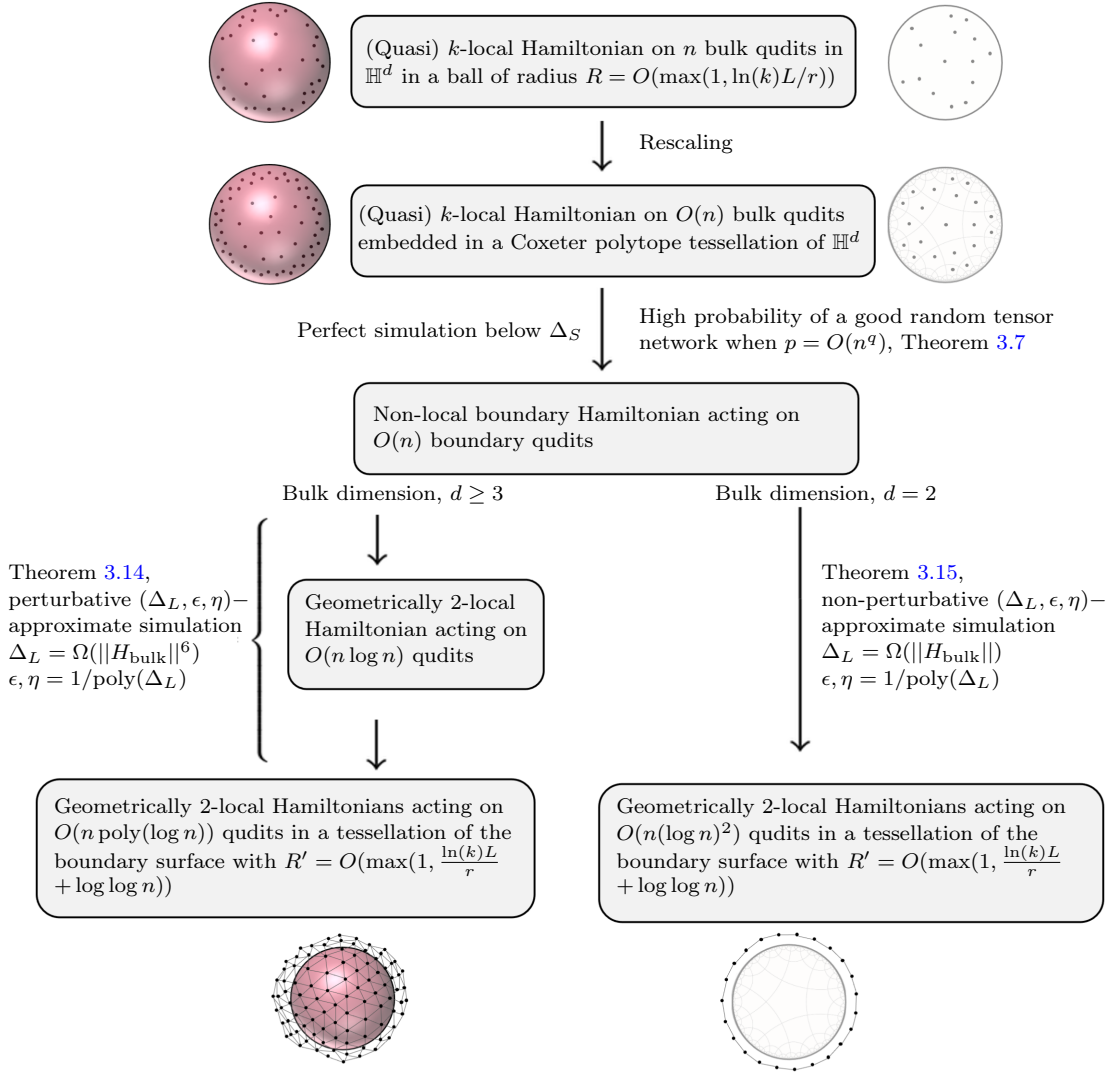

These `quasi-local Hamiltonians' allow for bulk Hamiltonians which contain Pauli rank-1 Wilson lines, which extend to the boundary of the tensor network.
We can now state our main results for the 3D/2D and 2D/1D dualities.
A simplified flow-diagram of our proof structure in both the 2d and 3d bulk cases is given in \cref{fg simulation flow chart} -- readers may find it helpful to get an overview of the proof structure from the diagram before delving into the technical details.

\begin{theorem}[Main result: 3D to 2D holographic mapping]
\label{thm Main result}
Let $\mathbb{H}^3$ denote 3D hyperbolic space, and let $B_r(x) \subset \mathbb{H}^3$ denote a ball of radius $r$ centred at $x$. Consider any arrangement of $n$ qudits in $\mathbb{H}^3$ such that, for some fixed $r$, at most $k$ qudits and at least one qudit are contained within any $B_r(x)$. Let $L$ denote the minimum radius ball $B_L(0)$ containing all the qudits (which wlog we can take to be centred at the origin). Let $H_\textup{bulk} = \sum_Z h_Z$ be any (quasi)-local Hamiltonian on these qudits.

Then we can construct a Hamiltonian $H_\textup{boundary}$ on a 2D boundary manifold $\mathcal{M}\in \mathbb{H}^3$
with the following properties:
\begin{enumerate}
\item $\mathcal{M}$ surrounds all the qudits, has radius $O\left(\max(1, \frac{\ln(k)}{r})L + \log \log n \right)$, and is homomorphic to the Euclidean 2-sphere.
\item The Hilbert space of the boundary consists of a tesselation of $\mathcal{M}$ by polygons of $O(1)$ area, with a qudit at the centre of each polygon, and a total of $O\left( n(\log n)^4 \right)$ polygons/qudits.
\item Any local observable/measurement $M$ in the bulk has a set of corresponding observables/measurements $\{M'\}$ on the boundary with the same outcome. Any local bulk operator $M$ can be reconstructed on a boundary region $A$ if $M$ acts within the entanglement wedge\footnote{\label{entanglement_foot}The entanglement wedge as defined by the spin domain wall in the corresponding spin picture as described in \cref{lm Mapping to an Ising partition function}} of $A$, denoted $\mathcal{E}[A]$.   This implies complementary recovery.
\item $H_\textup{boundary}$ consists of 2-local, nearest-neighbour interactions between the boundary qudits.
\item $H_\textup{boundary}$ is a $(\Delta_L, \epsilon, \eta)$-simulation of $H_\textup{bulk}$ in the rigorous sense of~\cite{simulation}, Definition 23, with $\epsilon, \eta = 1/ \textup{poly}(\Delta_L)$, $\Delta_L = \Omega \left(||H_\textup{bulk}||^6 \right)$, and where the maximum interaction strength $\Lambda = \max_{ij} |\alpha_{ij}|$ in $H_\textup{boundary}$ scales as $\Lambda = O\left( \Delta_L^{\textup{poly}(n\log (n))} \right)$.
\item The entanglement entropy of any subregion of the boundary agrees exactly with the Ryu-Takayanagi formula when there is no entanglement in the bulk.
\end{enumerate}
\end{theorem}
\begin{proof}
In this proof we use a tensor network construction where the network's underlying graph is a tessellation of 3-d hyperbolic space generated from Coxeter polytopes.
We place a random stabilizer tensor in each polyhedral cell of the finite bulk tessellation.
Each tensor has one index identified as the bulk qudit while the rest are contracted with tensors in neighbouring polyhedral cells.

In order to demonstrate that properties 1,2,4,5 and 6 hold we need that with high probability every tensor in our tensor network is simultaneously perfect, and that the network exactly obeys the Ryu-Takayanagi formula.
From \cref{thm Random stabilizer tensors are perfect} and \cref{lm Exact Ryu-Takayanagi} we have probability bounds on both these events occurring individually.
To get a bound the joint probability of these related events we again use Fr\'echet inequalities:
\begin{multline}\label{eqn probability scaling}
\text{Prob}\left[S(A)=S_{RT}(A) \cap \text{ perfect network}\right] \\
\geq \max \left\{ 0, \left[1- \frac{1}{2p^b}{t\choose \floor*{t/2}} \right]^n - \frac{1}{a} \right\}.
\end{multline}
This probability can be pushed arbitrarily close to 1 by increasing the bond dimension of the tensors.

For larger tensor networks, the bond dimension must be chosen larger in order to maintain high probability of having all perfect tensors. The bond dimension must scale as $p=O(n^{q})$ for $q>1/b$ in order for the probability of having an exact isometry tensor network to go to 1 as the size of the network increases.
However for any given $n$, $p$ can be chosen to be some large finite constant such that all properties resulting from tensor network being perfect will follow automatically for our modified HQECC.
So, with high probability we can construct a tensor network that is composed of perfect tensors which exactly obeys the Ryu-Takayanagi formula.

The local boundary Hamiltonian is built up by composing simulations using results from~\cite{simulation}.
The series of simulations are exactly the simulations used in Steps 1-3 in the proof of \cite[Theorem 6.10]{Tamara} -- we recap them here for completeness, but refer readers to the original for further details.\footnote{The extension of the proof in~\cite{Tamara} to cover the case of quasi-local Hamiltonians is trivial since the generalisation does not affect the scaling of weights of boundary operators -- see \cite[Section 6.1.5]{Tamara}.}

\paragraph{Step 1:} First we simulate the bulk Hamiltonian with a Hamiltonian that acts on the bulk indices of a HQECC in $\mathbb{H}^3$ of radius $R = O\left(\max(1, \frac{\ln(k)}{r})L  \right)$.\footnote{This scaling occurs since we may need to rescale the distances between qudits to embed them in a tessellation -- see \cite[Theorem 6.10]{Tamara} for details.}
This is a perfect simulation.
The new bulk Hamiltonian act on $O(n)$ qudits of dimension $p$, and the HQECC will contain $O(n)$ random stabilizer tensors.
Since the network is composed of perfect tensors with high probability, we can then define a non-local boundary Hamiltonian which is a perfect simulation of the original bulk Hamiltonian using the isometry defined by the HQECC.
This non-local boundary Hamiltonian also acts on $O(n)$ qudits.

\paragraph{Step 2}:
Next we simulate the non-local boundary Hamiltonian by a geometrically 2-local qudit Hamiltonian using the peturbative simulations outlined in  \cite[Theorem 6.10]{Tamara}.

\paragraph{Step 3}:
Finally we use the perturbative simulations from Step 3 in \cite[Theorem 6.10]{Tamara} to reduce the degree of each vertex to $3(p-1)$ so that the interaction graph can be embedded in a tessellation of the boundary surface.
\paragraph{}
The perturbative simulations used in Steps 2 and 3 are approximate, but the errors they introduce are tracked in \cite[Theorem 6.10]{Tamara} and can be made arbitrarily small by tuning the gadget parameters in the Hamiltonian. Practically, perturbation gadgets involve introducing ancillary qudits with neighbouring interactions into the Hamiltonian interaction graph. These `gadgets' modify the graph to obtain a new operator that reproduces the same low energy spectrum with different interactions.

In order to demonstrate the scaling claimed in the theorem we need to upper bound the number of ancillary qudits introduced by the perturbative simulations.
This requires us to determine the distribution of Pauli weights of terms in the boundary Hamiltonian after Step 1 of the simulation.
As long as the tensor network preserves Pauli rank of operators, this distribution will be unchanged from that calculated in \cite[Theorem 6.10]{Tamara}.
For the geometry in~\cite{Tamara,kohler2020translationallyinvariant} Pauli rank is preserved through the network since there exists a basis for the family of codes described by the perfect stabilizer tensor that maps logical Pauli operators to physical Pauli operators (Theorem D.4 in~\cite{Tamara}).

Any individual high-dimensional random stabilizer tensor is perfect with high probability and so a consistent basis exists for the family of codes described by any \textit{individual} tensor in our network. This is necessary since when considering a single tensor as an error correcting code from $1$ to $t-1$ legs, or as a code from $2$ to $t-2$ legs, some of the output legs are the same. Therefore for both codes to preserve Pauli rank there must be consistent basis for those legs. However unlike Kohler and Cubitt's construction, the basis that preserves Pauli rank will not be consistent across different tensors since every random stabilizer tensor will not be concentrated on the same perfect tensor. We can convince ourselves that this is acceptable by examining the hyperbolic geometry in \cref{fg holographic code} and seeing that the output legs of separate tensors' codes are always disjoint i.e. there's never a leg in the tensor network which is an output for two different tensors. The two sets of output legs are independent so the basis for the different tensors can be chosen independently.

Therefore the number of ancillary qudits introduced by our simulations will be unchanged from \cite[Theorem 6.10]{Tamara}, where it is shown that the final boundary Hamiltonian acts on $O(n \poly(\log(n)))$ qudits.
Since we enforce that the polygon cells of the boundary tessellation have area $O(1)$, the final boundary surface must have radius $R' = O\left(\max(1, \frac{\ln(k)}{r})L + \log \log n \right)$.

Properties 1,2,4, 5 and 6 follow immediately from the argument above and the results in \cite[Theorem 6.10]{Tamara}.
The complementary recovery described in point 3 follows from \cref{lm Complementary recovery}, since our construction can satisfy the conditions of the lemma, as described in \cref{Full complementary recovery}.
\end{proof}

It should be noted that in \cite[Theorem 6.10]{Tamara} the perturbative simulations are taken one step further, and the final boundary Hamiltonian is a qubit Hamiltonian with full local $SU(2)$ symmetry.
In~\cite{Tamara} this can be achieved without changing the scaling of the final boundary radius since in the perfect tensor case the local Hilbert dimension $p$ of the tensors is independent of the size of the tensor network.
This means that going from a Hamiltonian as outlined in \cref{thm Main result} to a qubit Hamiltonian with full $SU(2)$ local symmetry requires $O(1)$ rounds of perturbation theory, and $O(n)$ additional ancilla qubits, which does not affect the final scaling.

However, in the case of random tensors we need the ability to increase the local dimension of the tensors as the size of the tensor network increases, to ensure that \emph{every} tensor in the network is perfect with high probability.
The bond dimension must scale as $p=O(n^{q})$ for $q>1/b$ in order for the probability of having an exact isometry tensor network to go to 1 as the size of the network increases (where $b < 1$ as defined in \cref{thm Random stabilizer tensors are perfect}).
This dependence of $p$ on $n$ means that reducing the boundary Hamiltonian to act on qubits and have $SU(2)$ symmetry requires $O(n^{1+2q}\poly(\log( n))$ ancilla qudits.
Therefore maintaining the density of qubits on the boundary surface would require a boundary radius of $R' = O((1+2q)L') > 3L'$, where $L' = \max(1, \frac{\ln(k)}{r})L$ and $L$ is the radius that contains all the bulk qudits.
At this point the boundary can no longer be considered a geometric boundary of the bulk geometry, which is why we omit the final simulation step from our construction.

It may be possible to construct a boundary Hamiltonian on qudits with $SU(2)$ or $SU(p)$ local symmetry while maintaining the scaling outlined here by using qudit perturbation gadgets.
It is known that the qudit generalisations of the Heisenberg Hamiltonian with either symmetries remain universal for all local dimension greater than 2 (see~\cite{piddock2018universal}), but the scaling of the ancillas required would need to be investigated.

\begin{theorem} [Main result: 2D to 1D holographic mapping] \label{holography}
Consider any arrangement of $n$ qudits in $\mathbb{H}^2$, such that for some fixed $r$ at most $k$ qudits and at least one qudit are contained within any $B_r(x)$.
  Let $Q$ denote the minimum radius ball $B_Q(0)$ containing all the qudits.
  Let $H_\textup{bulk} = \sum_Z h_Z$ be any (quasi) $k$-local Hamiltonian on these qudits.

  Then we can construct a Hamiltonian $H_\textup{boundary}$ on a 1D boundary manifold $\mathcal{M}$ with the following properties:
  \begin{enumerate}
  \item
    $\mathcal{M}$ surrounds all the qudits and has radius $\BigO\left(\max\left(1,\log(k)/r\right) Q + \log\log n\right)$.
      \item
    The Hilbert space of the boundary consists of a chain of qudits of length $\BigO\left(n\log n \right)$.
  \item
 Any local observable/measurement $M$ in the bulk has a set of corresponding observables/measurements $\{M'\}$ on the boundary with the same outcome. Any local bulk operator $M$ can be reconstructed on a boundary region $A$ if $M$ acts within the entanglement wedge\cref{entanglement_foot} of $A$, denoted $\mathcal{E}[A]$.   This implies complementary recovery.
  \item
    $H_\textup{boundary}$ consists of 2-local, nearest-neighbour interactions between the boundary qudits.
  \item
    $H_\textup{boundary}$ is a $(\Delta_L,\epsilon,\eta)$-simulation of $H_\textup{bulk}$ in the rigorous sense of~\cite{simulation}, with $\epsilon,\eta = 1/\poly(\Delta_L)$, $\Delta_L = \Omega\left(\|H_\text{bulk}\|\right)$, and where the interaction strengths in $H_\textup{boundary}$ scale as  $\max_{ij}|\alpha_{ij}| = \BigO\left(\Delta_L\right)$.
    \item
     The entanglement entropy of any subregion of the boundary agrees exactly with the Ryu-Takayanagi formula when there is no entanglement in the bulk.

  \end{enumerate}
\end{theorem}

\begin{proof}
In this proof we use a tensor network construction where the network's underlying graph is a tessellation of 2-d hyperbolic space generated from Coxeter polytopes.
We place a random stabilizer tensor in each polyhedral cell of the finite bulk tessellation.
Each tensor has one index identified as the bulk qudit while the rest are contracted with tensors in neighbouring polyhedral cells.
The proof that with high probability every tensor in the network is perfect, and that the network obeys the Ryu-Takayanagi formula follows exactly as in \cref{thm Main result}.

The first step in the simulation is unchanged from \cref{thm Main result}.
Then, instead of using perturbative simulations to construct the local boundary Hamiltonian, we use the history-state simulation method, as in \cite[Theorem 5.2]{kohler2020translationallyinvariant}.
The history state simulation method is again approximate, but it is non-perturbative.
The overhead it incurs both in terms of ancillary qudits and errors is calculated in \cite[Theorem 5.2]{kohler2020translationallyinvariant} -- as in  \cref{thm Main result} we can use the fact that our random tensor networks preserve the Pauli rank of operators to argue that the scaling here is unchanged from the perfect tensor case.
Again, we have to take care since the dimension of the qudits now scales as $p=O(n^q)$.
This means that describing a Pauli rank-1 operator of weight $w$ requires $O(w\log(n))$ bits of information, as opposed to $O(w)$ in the perfect tensor case.
So the number of spins on the boundary manifold scales as $O(n\log(n)^2)$.
Due to the hyperbolic geometry this does not change the asymptotic scaling of the distance from the new boundary to the old boundary.

Properties 1,2,4, 5 and 6  follow immediately.
As in \cref{thm Main result} the complementary recovery described in point 3 follows from \cref{lm Complementary recovery}.

\end{proof}

It also follows from the above that all the work in~\cite{Tamara}: exploring the map between models, the dynamics of the bulk and boundary and their relative energy scales, also applies to our modified holographic code construction with high probability. Our construction illustrates that a mapping between models is possible without having to chose carefully selected tensors, which allows us to simultaneously achieve a mapping between models, perfect Ryu-Takayanagi and perfect complementary recovery.

\section{Conclusions and outlook}
\label{Conclusions and outlook}

Quantum codes with a tensor network structure have provided exactly solvable toy models for several interesting holographic properties of the AdS/CFT correspondence. Independent models have succeeded on different fronts: averages of random tensor networks can be mapped to spin systems giving the Ryu-Takayanagi entropic relation for general cases, whereas a duality at the level of Hamiltonians giving insights into dynamical features and energy scales has been achieved using simulation techniques on perfect tensor networks. Tensor network constructions of holographic codes have been a fruitful area of research leading to classification and study of convenient tensors. This work outlines a mathematically rigorous characterisation of the concentration of random stabilizer tensors about perfect tensors with increasing bond dimension. Exploiting the algebraic structure of the stabilizer group we obtain a probability bound on having an exact perfect tensor.

We proposed new HQECCs based on the constructions in~\cite{Tamara} and~\cite{kohler2020translationallyinvariant} but replacing the perfect tensors in their networks with high dimensional random stabilizer tensors. By increasing the dimension of the qudits describing the tensor legs, we are able to construct a tensor network from random stabilizer tensors that inherits all the desirable features of a perfect tensor network with high probability. The simulation techniques from~\cite{Tamara,kohler2020translationallyinvariant} used to achieve a mapping between models from perfect tensors can be analogously applied to a network of random stabilizer tensors retaining the same scaling of boundary qudits and interaction strengths with high probability.
The successes of Kohler and Cubitt's codes are all inherited by these modified HQECCs: they are exactly solvable; have error correcting properties realising the proposal in~\cite{Rindler}; have a uniform bulk and are able to explore the dynamical features of the duality and energy scales. Demonstrating these properties without perfect tensors is useful since it removes the overhead of explicitly checking the existence and constructing a special tensor in the given dimension.

An important advantage of replacing perfect tensors with random stabilizer tensors is the improved entanglement features. We demonstrate exact Ryu-Takayanagi for general boundary configurations for random stabilizer tensor networks advancing from the singly connected regions in perfect tensor networks. This is shown rigorously for product bulk states, but following earlier work qualitative statements can be made regarding the expected corrections when entanglement is present in the bulk. Our construction does not suffer the exceptions to complementary recovery found in the HaPPY code, Kohler and Cubitt's construction, and general random HQECCs. This new construction exhibits complementary recovery for any bulk index and any bipartition of the boundary.

Our HQECC takes a further step towards a complete, mathematically rigorous construction of holographic duality, simultaneously capturing more features of the AdS/CFT correspondence than previous work. We are still however a long way from a complete mathematical description of AdS/CFT. The boundary theory arising from our encoding is not Lorentz invariant in the limit and is constructed in Euclidean rather than Minkowski space. An interesting avenue for further work could be to attempt to achieve a discrete version of conformal invariance in the boundary model. Work by Osborne et al. in~\cite{HQECC2} also extended the HaPPY code to introduce dynamics via a different method which led to a boundary system with conformal invariance, although the redundant encoding of information was not replicated. These results could potentially be further strengthened by combining the successful techniques from~\cite{HQECC2} into our construction.

As with previous HQECC constructions, the holographic states in our construction all have flat Renyi entanglement spectra.
In the past it has been suggested that this is a drawback of these models, since this is not the entanglement spectrum predicted from general relativity~\cite{2016}.
However, recent work has demonstrated that the `fixed area states' of quantum gravity (which are states corresponding to tensor network states), do in fact have flat entanglement spectra~\cite{2019}.
The more complicated entanglement spectra of general states can be understood by considering superpositions over states with fixed area, or, in the HQECC picture, superpositions over different tensor network geometries~\cite{2017}.
Another avenue to explore is how to introduce these superpositions dynamically into the construction so the expected entanglement spectrum emerges naturally.

\section*{Acknowledgements}
The authors would like to thank Sepehr Nezami for very helpful responses regarding random tensor network constructions.
The authors are grateful to Michael Walter for helpful initial discussions and pointing out a gap in the proof concerning concentration of random stabilizer tensors.
T.\,S.\,C.~is supported by the Royal Society.
T.\,K.~is supported by the EPSRC Centre for Doctoral Training in Delivering Quantum Technologies [EP/L015242/1].
H.\,A. ~is supported by EPSRC DTP Grant Reference: EP/N509577/1 and EP/T517793/1.

\begin{appendices}
\crefalias{section}{appendix}

\section{Lipschitz constant for $\trace(\sigma_A^2)$}
\label{Lipschitz constant}

In order to apply Levy's lemma to the purity of Haar random states we need to upper bound the function's Lipschitz constant. This will give a quantified measure of the `smoothness' of the function. The Lipschitz function for purity is simple to compute and found in previous works, e.g. see~\cite{Low_2009} Lemma 4.2. We include a short proof here for completeness.

\begin{lemma}[Lipschitz constant for purity]
\label{lm purity Lipschitz}
The Lipschitz constant for purity is upper bounded by $\eta \leq 2$.
\end{lemma}
\begin{proof}
The purity is equal to the square of the 2-norm, $\trace(\rho^2)= \sum_{i=1}^n \lambda_i^2= ||\rho||_2^2$. Therefore the Lipschitz constant of purity is given by
\begin{subequations}
\begin{align}
\eta &= \sup_{\rho_1,\rho_2} \frac{\left| ||\rho_1||_2^2 - ||\rho_2||^2_2\right|}{||\rho_1-\rho_2||_2}\\
&= \sup_{\rho_1,\rho_2} \frac{\left| ||\rho_1||_2 - ||\rho_2||_2\right|(||\rho_1||_2 + ||\rho_2||_2)}{||\rho_1-\rho_2||_2}.
\end{align}
\end{subequations}
Since the 2-norm obeys the reverse triangle inequality,
\begin{subequations}
\begin{align}
\eta &\leq \sup_{\rho_1,\rho_2} \frac{\left| ||\rho_1||_2 - ||\rho_2||_2\right|(||\rho_1||_2 + ||\rho_2||_2)}{\left| ||\rho_1||_2 - ||\rho_2||_2\right|}\\
&= \sup_{\rho_1,\rho_2} (||\rho_1||_2 + ||\rho_2||_2).
\end{align}
\end{subequations}
Finally since $||\rho||_2\leq1$ for all $\rho$ we arrive at the result.
\end{proof}

\section{Low temperature Ising model corrections}
\label{Low temperature Ising model corrections}

The ground state of the Ising model in $>1$ dimensions is given by the number of tensor `legs' crossed by the minimal domain wall. The energy gap, $\Delta$, between the ground and first excited state comes from one additional tensor leg being crossed by the domain wall, giving energy penalty $\Delta =1$. Writing out the sum,
\begin{subequations}
\begin{align}
Z &= e^{-\beta E_{GS}} \left(1 + e^{-\beta} + e^{-2\beta} +... \right)\\
&= e^{-\beta E_{GS}}(1+\epsilon).
\end{align}
\end{subequations}
Calculating $\epsilon$ as the sum of a GP:
\begin{subequations}
\begin{align}
\epsilon &= \sum_{i=1}^N e^{-\beta} (e^{-\beta})^{i-1}\\
& = \frac{e^{-\beta}(1-e^{-N\beta})}{1-e^{\beta}} = \frac{\frac{1}{p}\left(1 - \frac{1}{p} \right)}{1- \frac{1}{p}} = \frac{p^N - 1}{p^{N+1}-1} = O\left(\frac{1}{p} \right).
\end{align}
\end{subequations}
Using the expansion of $Z$:
\begin{subequations}
\begin{align}
-\ln Z &=  - \ln \left[e^{-\beta E_{GS}}\right] - \ln \left[ (1+\epsilon) \right] \\
&= - \ln \left[e^{-\beta E_{GS}}\right] - \epsilon - \frac{\epsilon^2}{2} + \frac{\epsilon^3}{3}-... \tag*{using the Taylor expansion of $\log(1+\epsilon)$}\\
&= \beta E_{GS} - O\left(\frac{1}{p} \right),
\end{align}
\end{subequations}
with corrections vanishing with increasing bond dimension. This result holds for any graph in $>1$ dimensions. However if there are $k$ degenerate minimal surfaces through the tensor network $-\ln Z$ is modified by $-\ln k$.

\section{Generalisation of the Ising mapping}
\label{Generalisation of the Ising mapping}

Section 4 of~\cite{Random} introduces a generalisation of the Ising mapping where the R\'enyi entropy considered is a function of a state with different support. We now take the full tensor network state, written as a pure state defined by the tensor network isometry, $V$, in an orthonormal basis of the bulk, $\{\ket{\alpha} \}$, and boundary, $\{ \ket{a}\}$:
\begin{equation}
\ket{\Psi_V} = \frac{1}{p^{v/2}} \sum_{\alpha, a} \braket{\alpha|V|a} \ket{\alpha} \otimes \ket{a},
\end{equation}
and trace out subregions of either the bulk or the boundary or both. $v$ here is the number of vertices in the tensor network graph.

Similarly to the mapping described for boundary states, they find that a function related to the R\'enyi-2 entropy of $\rho = \ket{\Psi_V}\bra{\Psi_V}$ can be equated to an Ising model partition function:
\begin{equation}
\langle \trace \left[\left(\rho \otimes \rho \right)\mathcal{F}_{Y}\right] \rangle = \sum_{\{s_x \}} e^{-\mathcal{A}[\{s_x \}]},
\end{equation}
where $Y$ can be any subregion of the full tensor network state i.e. include part of the boundary and/or the bulk. The Ising action is given by
\begin{equation}\label{eqn appen F action}
\mathcal{A}\left[\{ s_x\} \right] = -\frac{1}{2} \log D \left[ \sum_{\langle xy \rangle}(s_xs_y-1) + \sum_{x\in \mathcal{B}}(h_xs_x-1) \right] - \frac{1}{2} \log D_b \sum_{x\in b} (b_x s_x -1),
\end{equation}
where we have introduced a bulk pinning field $b_x$,
\begin{align}
\text{Boundary pinning field, } h_x &= \begin{cases}
+1 & x \in Y\\
-1 &  x \in \bar{Y}
\end{cases}\\
\text{Bulk pinning field, } b_x &= \begin{cases}
+1 & x \in Y\\
-1 &  x \in \bar{Y}
\end{cases}
\end{align}
We have also introduced $D$, the dimension of the bond connecting two tensors in the network and $D_b$, the dimension of the bulk dangling index.
In our construction we chose each leg of our tensors to have dimension $p$ and achieve different dimensions for the links within the network and the bulk degrees of freedom by grouping legs together in the faces of the polytopes.

By the same method as \cref{Low temperature Ising model corrections} in the limit of large bond dimension the free energy of this Ising model (defined without the temperature prefactor) is dominated by the Ising model ground state:
\begin{equation}
-\ln \langle \trace \left[\left(\rho \otimes \rho \right)\mathcal{F}_{Y}\right] \rangle = |\gamma_Y| \ln p - \ln k - O(1/p).
\end{equation}
Where $|\gamma_Y|$ is the length of the domain spin wall in the ground state of the Ising model defined by \cref{eqn appen F action}.
As before, this convenient function is not directly the R\'enyi-2 entropy,
\begin{equation}
S_2(\rho_Y) = -\frac{1}{2} \ln \trace \left[\rho_Y \right] = \frac{1}{2} \ln \trace \left[\left(\rho \otimes \rho  \right)\mathcal{F}_Y\right]
\end{equation}
where $\rho_Y = \trace_{\bar{Y}}(\rho)$.
However, the method of appendix F of~\cite{Random} follows through when considering a full tensor network state rather than a boundary state so we can again bound the average R\'enyi-2 entropy in terms of the ground state energy of the Ising model
\begin{equation}
|\gamma_Y| \ln p - \langle S(\rho_Y) \rangle \leq \ln k + o(1).
\end{equation}
Since $\ket{\Psi_V}$ is a stabilizer state and hence has quantised entropy, following the proof of \cref{lm Approximate Ryu-Takayanagi} we can conclude that the von Neumann entropies of subregions of this state can be made equal to $|\gamma_Y| \ln p$ with probability greater than $\left(1-\frac{1}{\delta} \right)$ which can be made arbitrarily close to 1 by choosing a sufficiently high bond dimension $p$.

\end{appendices}

\bibliographystyle{JHEP}
\bibliography{RandomTensorsJHEP}

\end{document}